\def\full{}

\PassOptionsToClass{hyphens}{url}
\ifdefined\full
\documentclass[sigconf,nonacm,urlbreakonhyphens]{acmart}
\else
\RequirePackage{pdf14}
\documentclass[sigconf,urlbreakonhyphens]{acmart}

\copyrightyear{2020}
\acmYear{2020}
\setcopyright{acmlicensed}
\acmConference[CCSW '20]{2020 Cloud Computing Security Workshop}{November 9, 2020}{Virtual Event, USA}
\acmBooktitle{2020 Cloud Computing Security Workshop (CCSW '20), November 9, 2020, Virtual Event, USA}
\acmPrice{15.00}
\acmDOI{10.1145/3411495.3421362}
\acmISBN{978-1-4503-8084-3/20/11}
\fi

\usepackage{rotwein}
\hypersetup{breaklinks=true}

\title{Short-Lived Forward-Secure Delegation for TLS}
\ifdefined\full
\titlenote{This is the full version of a paper which appears in 2020 Cloud Computing Security Workshop (CCSW '20), November 9, 2020, Virtual Event, USA, ACM. \url{https://doi.org/10.1145/3411495.3421362}}
\fi

\author{Lukas Alber}
\email{lukas.alber@iaik.tugraz.at}
\affiliation{%
  \institution{Graz University of Technology}
  \city{Graz}
  \country{Austria}
}

\author{Stefan More}
\orcid{0000-0001-7076-7563}
\email{stefan.more@iaik.tugraz.at}
\affiliation{%
  \institution{Graz University of Technology}
  \city{Graz}
  \country{Austria}
}

\author{Sebastian Ramacher}
\orcid{0000-0003-1957-3725}
\email{sebastian.ramacher@ait.ac.at}
\affiliation{%
  \institution{AIT Austrian Institute of Technology}
  \city{Vienna}
  \country{Austria}
}


\begin{document}
\ifdefined\full\else
\fancyhead{}
\fi


\begin{abstract}
On today's Internet, combining the end-to-end security of TLS with Content Delivery Networks (CDNs) while ensuring the authenticity of connections results in a challenging delegation problem.
When CDN servers provide content, they have to authenticate themselves as the origin server to establish a valid end-to-end TLS connection with the client.
In standard TLS, the latter requires access to the secret key of the server.
To curb this problem, multiple workarounds exist to realize a delegation of the authentication.

In this paper, we present a solution that renders key sharing unnecessary and reduces the need for workarounds.
By adapting identity-based signatures to this setting, our solution offers short-lived delegations.
Additionally, by enabling forward-security, existing delegations remain valid even if the server's secret key leaks.
We provide an implementation of the scheme and discuss integration into a TLS stack.
In our evaluation, we show that an efficient implementation incurs less overhead than a typical network round trip.
Thereby, we propose an alternative approach to current delegation practices on the web.
\end{abstract}

\maketitle

\section{Introduction}
Transport Layer Security (TLS)~\cite{DBLP:journals/rfc/rfc8446} is the primary protocol for end-to-end secure communication between two parties on untrusted networks -- and most importantly -- on the Internet. It provides confidentiality and authenticity of the transmitted payloads. Authenticity of the endpoints and especially of the server is usually established via certificates managed by a Public Key Infrastructure (PKI) in the initial phase of the protocol, the handshake. In the process, at least the server presents its certificate to the other party for authentication.\footnote{Other methods such as pre-shared keys are not of importance for this work.}
Certificates themselves are digitally signed documents that bind a public key and validity information to an identity, e.g., to a hostname.
In the PKI deployed on the web, certificates are signed by trusted third-parties called Certificate Authorities (CAs) vouching for their validity.

To meet the demand for performance, scalability and security, Content Delivery Networks (CDNs) are now widely deployed.
CDNs provide content mirrored from the origin server via surrogate servers closer to the client.
In other words, they are systems of globally distributed servers that deliver content on behalf of webservers to users.
For popular applications, a CDN reduces the load on the origin servers, i.e., the webserver of the web application, and allows applications to scale up to some extent as they grow popular.
A survey performed by Cisco Systems from 2017 until 2018 estimates that the data volume of the global CDN traffic will reach 252 exabytes per month in 2022~\cite{cdnstat}.

 \begin{figure}[t]
    \centering
    \includegraphics[width=\columnwidth]{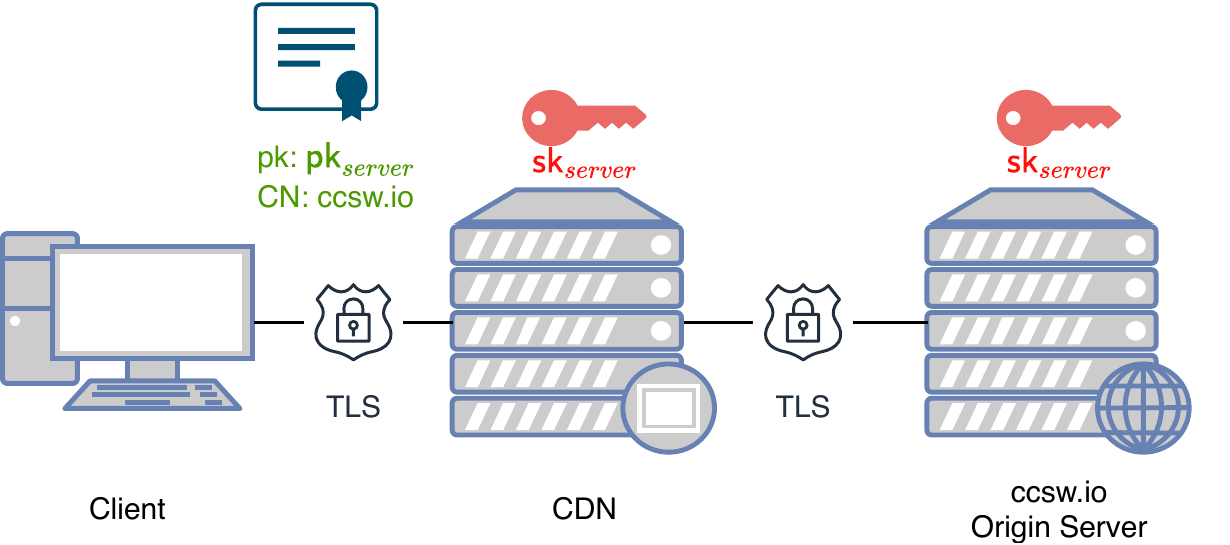}
    \caption{Example network architecture of an origin server using a Content Delivery Network (CDN). All communication over the Internet is secured using TLS. The server grants the right to act in its name to the CDN by sharing the secret key \skServer.}
    \label{fig:arch}
\end{figure}

Since security requirements do not change when using a CDN, CDNs need to serve an origin server's content over TLS.
Further, as the web client needs to verify the identity of the server using a certificate, the secret key counterpart of the public key is needed to sign the TLS handshake for authentication.
Thus, for a CDN to share content on behalf of an origin server,
the CDN needs access to the secret key of the certificate bound to the origin server's domain.
Since TLS -- by design -- currently lacks any form of delegation, the easiest solution is to simply hand over the secret key to the CDN.
However, sharing the secret key constitutes a security loophole as the owner of the domain loses control over the associated certificate to a CDN.
It makes more entities privy to the secret key of a certificate and consequently increases the possibility of key leakage.
Furthermore, an origin server cannot constrain a CDN to only respond to specific requests nor constrain a CDN from acting without permission on its behalf. The delegation can also not be revoked without revoking the origin server's certificate and rotating keys. Suffice to say, such practices increases security risks in the TLS ecosystem, as discussed in recent works~\cite{DBLP:conf/ccs/CangialosiCCLMM16,DBLP:conf/sp/LiangJDLWW14}.

The practice of secret key sharing as illustrated in \Cref{fig:arch} is also known as Custom Certificate~\cite{DBLP:conf/sp/LiangJDLWW14}. Together with Cruiseliner Certificates (or Shared Certificates)~\cite{DBLP:conf/ccs/CangialosiCCLMM16,DBLP:conf/sp/LiangJDLWW14} these practices are known as ``full delegation''~\cite{DBLP:conf/ccs/MamboUO96,DBLP:journals/joc/BoldyrevaPW12} in the literature. Cruiseliner Certificates allow CDNs to manage the users' certificates on their behalf. That is usually done with certificates valid for multiple domains using the Subject Alternative Name X.509 extension~\cite{DBLP:journals/rfc/rfc4985}. With these types of certificates, the domain owners also loose control over their certificates.
Additionally, listing a set of domains in the same certificate has implication on privacy, e.g., when it reveals internal hostnames.
Furthermore, it can enable customers, served using the same certificate, to impersonate each other.

To give an estimation on how widespread such practices are in practice, Cangialosi et al.~\cite{DBLP:conf/ccs/CangialosiCCLMM16} analyzed their use in 2016 and reported that 76.5\% of all organizations on the web share at least one private key with third-party organizations such as CDNs.
Some of these third-party organizations are responsible for certificate management, i.e., revocation and reissuing.
Furthermore, ten of such third-party organizations have the private keys of as much as 45.3\% of all the sites on the web.

\paragraph{Delegation in TLS}
Alternative approaches to avoid the security issues of full delegation have been proposed over the years (cf. \cite{chuat2019sok} for a recent survey of delegation techniques). We will discuss some approaches in the following.
Liang et al.~\cite{DBLP:conf/sp/LiangJDLWW14} carried out a systematic study on the interplay between HTTPS and CDNs in 2014.
On examination of 20 popular CDN providers and 10,721 of their customers, they revealed various problems regarding practices adopted by CDN providers, such as private key sharing, neglected revocation, and insecure back-end communication.
As a remedy, they proposed a lightweight extension for DANE~\cite{DBLP:journals/rfc/rfc6698} to tackle the delegation problem caused by the end-to-end nature of HTTPS and the man-in-the-middle nature of CDNs.
In short, the origin server adds both its certificate and the CDNs certificate as its TLSA records to DNS.
On receiving the CDN's certificate, the client recognizes the delegation if the CDN's certificate appears in the origin server's TLSA record.
However, the proposed solution by Liang et al. requires changes to certificate validation on the client.
Besides, the extra roundtrip to retrieve the TLSA records during a TLS handshake increases connection latency.

Cloudflare deployed Keyless SSL~\cite{cloudflare-keyless-ssl-details} to omit sharing of the private keys.
In short, Keyless SSL splits up the TLS handshake. It allows to establish a TLS session to the CDN, while the private key operations are outsourced to a key server (controlled by the domain owner).
Therefore, the domain owner can maintain control of its private key and avoids full delegation of the key to the CDN operator. Stebila et al.~\cite{DBLP:conf/trustcom/StebilaS15} examined the security and performance of the approach. They found that performance slightly worsens.
Bhargavan et al.~\cite{DBLP:conf/eurosp/BhargavanBFOR17} demonstrated several new attacks on Keyless SSL.
They also presented a security definition for authenticated and confidential channel establishment involving 3 parties to model the additional key server, dubbed 3(S)ACCE-security, based on 2-party ACCE security definitions that have been used in several proofs for TLS~\cite{DBLP:conf/crypto/KrawczykPW13,DBLP:conf/crypto/JagerKSS12}.
Furthermore, they proposed modifications to Keyless SSL for TLS 1.2 achieving 3(S)ACCE-security guarantees.
Also, they strongly argued for a new design for Keyless TLS 1.3 which is computational lighter and requires simpler assumptions than the proposed modifications of TLS 1.2.

Recently, Delegated Credentials (DeC)~\cite{DelegatedCredentials} have been proposed to replace Keyless SSL.
A DeC is a digitally signed data structure that contains a validity period, and a public key along with the signature algorithm.
By signing the data structure, the origin server delegates to the included public key.
Integration into the TLS protocol happens via a specific TLS extension.
DeC appear to be an attractive proposition and are currently evaluated by Mozilla~\cite{mozilla-dec}, Cloudflare~\cite{cloudflare-dec}, and Facebook~\cite{facebook-dec}.

Chuat et al.~\cite{DBLP:journals/corr/abs-1906-10775} investigated the potential of proxy certificates solving the long-standing problems of revocation and delegation in the web PKI.
In this case, holders of a non-CA certificate can issue proxy certificates.
Thereby, holders of a non-CA certificate are able to delegated signing privileges to other entities in a fine-grained manner.
In case of a key compromise, Chuat et al. claim that limiting the proxy certificate's lifetime to an arbitrary short timespan may curb problems of revocation.

Name Constraints is an extension to X.509 certificates~\cite{DBLP:journals/rfc/rfc5280} which defines the namespace for which all subsequent certificates below the certificate that defines the Name Constraints extension must reside.
Since Name Constraints can only be used in a CA certificate, a root CA can issue an intermediate CA certificate to the origin server that allows the origin server to issue child certificates limited to hostnames in its restricted namespace.
However, Liang et al.~\cite{DBLP:conf/sp/LiangJDLWW14} showed that name constraints have poor support in browsers, and when not properly checked, allow the origin server to create and use certificates that are not restricted to its namespace.

In a different work~\cite{chuat2019sok}, Chuat et al. proposed a 19-criteria framework for characterizing those revocation and delegation schemes.
Further, they propose that combining short-lived DeC or proxy certificates with functional revocation may curb several of these problems in web PKI.

Finally, we want to mention approaches such as STYX~\cite{DBLP:conf/cloud/WeiLLYG17} which require the presence of trusted execution environments such as Intel SGX. Due to wide range of attacks against such systems, e.g., \cite{DBLP:conf/uss/BulckMWGKPSWYS18,DBLP:conf/ccs/0001LMBS0G19,DBLP:conf/sp/BulckM0LMGYSGP20}, we do not consider them as a viable approach.

\paragraph{Delegation and Multi-Entity Communication in Other Protocols}
Various types of delegations are also interesting for other protocols besides TLS or variants of TLS for special use-cases. Kogan et al. developed \emph{Guardian Angel}, a delegation agent, to solve the problem of delegation in SSH~\cite{DBLP:conf/hotnets/KoganSTMW17}.
The system allows a user to choose which delegate systems can run commands on which servers.

Naylor et al.~\cite{DBLP:conf/sigcomm/NaylorSVLBLPRS15} developed \emph{mcTLS}, which extends TLS to support middleboxes.
\emph{mcTLS} allows TLS endpoints (i.e., clients and webservers) to introduce middleboxes in secure end-to-end connection while restricting what the middleboxes can read or write.
However, Bhargavan et al.~\cite{DBLP:conf/sp/BhargavanBDFO18} show that \emph{mcTLS} is insecure and susceptible to a class of attack called \emph{middlebox confusion} attacks.
They suggested a provable secure alternative to \emph{mcTLS} avoiding the middlebox confusion issue.
Cho et al.~\cite{DBLP:conf/iotdi/ChoPLCK19} proposed \emph{D2TLS}, a mutual authentication agent that helps cloud-based IoT devices set up secure connections by leveraging the session resumption in DTLS~\cite{DBLP:journals/rfc/rfc6347} while keeping the device's secret key secure.
\emph{D2TLS} requires a change to the IoT device, and it is only capable of mutual authentication.

Wagner et al.~\cite{DBLP:conf/openidentity/WagnerOM17} discuss the application of delegation schemes in trust management. They propose a scheme similar to DeC, which uses a XML structure to define generic constraints of the delegation.

\paragraph{Delegation in Cryptography}
While standard signature and public-key encryption schemes are not designed to allow delegation of any kind, the introduction of identity-based cryptography~\cite{DBLP:conf/crypto/Shamir84,DBLP:conf/asiacrypt/GentryS02,DBLP:conf/sacrypt/Hess02} started an exciting avenue that enables delegation of signing and encryption rights. While literature usually formulates identity-based systems without public keys, the public parameters are often interpreted as a public key. From this viewpoint, (hierarchical) identity-based systems enable delegation based on the identities and can be managed per public key. For encryption, such fine-grained control is, for example, achieved by attributed-based encryption~\cite{DBLP:conf/ccs/GoyalPSW06}, functional encryption~\cite{DBLP:journals/cacm/BonehSW12}, and fully puncturable encryption~\cite{DBLP:journals/iacr/DerlerRSS19} built on ideas found in hierarchical identity-based encryption schemes such as the one of Boneh, Boyen and Goh~\cite{DBLP:conf/eurocrypt/BonehBG05}.

Interestingly, there is a close connection between identity-based encryption and signatures. Naor described and Boneh and Franklin later sketched~\cite{DBLP:conf/crypto/BonehF01} a transformation that turns secure identity-based encryption scheme into an unforgeable signature scheme. This observation can be generalized to (hierarchical) identity-based signature schemes: any hierarchical identity-based encryption~\cite{DBLP:conf/asiacrypt/GentryS02} scheme with $\ell+1$ levels implies a hierarchical identity-based signature scheme with $\ell$ levels (cf. \cite{DBLP:series/ciss/KiltzN09}). Alternatively, identity-based signatures can also be obtained generically from two standard signature schemes~\cite{DBLP:conf/crypto/Shamir84,DBLP:conf/eurocrypt/BellareNN04}. This paradigm can also be extended to obtain forward-secure identity-based signatures and identity-based proxy signatures~\cite{DBLP:conf/asiacrypt/GalindoHK06}.

The latter, proxy signatures, were introduced by Mambo et al.~\cite{DBLP:conf/ccs/MamboUO96}. Their goal is also to allow the delegation of signing rights. There, the verifier checks that the signature was signed by the proxy and that signing rights were delegated to the proxy. Similar to identity-based signatures, they can be built from standard signature schemes~\cite{DBLP:journals/joc/BoldyrevaPW12}. While proxy signatures have a long history, including paring- and lattice-based constructions~\cite{DBLP:conf/incos/Li16,DBLP:conf/secrypt/BuccafurriSS16}, they have not seen much adaption in the context of TLS. As noted before, Chuat et al.~\cite{chuat2019sok} propose them as one approach to fix the delegation issue in TLS.

\subsection{Our Contribution}
The goal of our research is to address delegation between the origin server and the CDN. A good solution to the problem should fulfill the following requirements:
\begin{enumerate*}
  \item Private keys must be kept private to the origin server, and the origin server must retain control over the private key. Thereby, compromise of keys outside of the origin server's control is mitigated.
  \item The origin server should be able to delegate signing privileges in a fine-grained manner. That is, regardless of the possible choices of CDNs, the origin server is able to delegate keys on a per CDN basis, where one delegation does not prohibit further delegations.
  \item The origin server should be able to decide the validity period of the delegation. Even if the origin server prefers short validity periods, the overhead of the delegation should stay minimal.
  \item Furthermore, the origin server's private key should be equipped with forward-security features to reduce the overall risk. Thereby, even if the server's private key leaks, delegations that occurred before the leak remain valid.
  \item Finally, support for such a system should not require specific hardware such as trusted execution environments and Intel SGX to be present.
\end{enumerate*}

In this work, we provide a solution following the proposal of Chuat et al.~\cite{chuat2019sok} on a cryptographic level.
In contrast to the emerging paradigm of Delegated Credentials (DeC), our approach does not delegate to a different key pair, but instead derives constrained secret keys for the existing key pair.
To do so, we investigate variants of identity-based signatures which we equip with additional features.
Specifically, we extend them with an epoch that fixes identity-based delegations to a specific time-span, forming a scheme called \emph{time-bound identity-based signatures} (\TBIDS).
Thereby, the origin server delegates signing rights with respect to the origin server's public key to the CDN. However, the key is only valid for a certain time period.
If the origin server later decides that the CDN should no longer be able to establish TLS connections using the origin server's certificate, it suffices to no longer provide delegated keys of forthcoming epochs to the CDN.
It is not necessary for the origin server to entirely revoke keys and obtain fresh certificates. Short-lived delegations therefore represent an alternative approach to revocation~\cite{topalovic2012towards,nirsaagstar}.

Furthermore, we provide a forward-secure version of time-bound identity-based signatures.
Contrary to forward-secure identity-based signatures, delegated keys stay bound to a specific time period.
In our case, forward-security is provided for the master secret key kept by the origin server by applying a technique introduced by Canetti, Halevi and Katz~\cite{DBLP:conf/eurocrypt/CanettiHK03}.

To underline the practicability of our approach, we provide implementations and benchmarks of the proposed schemes.
Specifically, we show that while usage of the scheme would incur a small overhead, the runtime overhead of an optimized implementation is below typical network round trip times.
Therefore, the performance impact on TLS is small.
This observation is confirmed by an integration of the scheme into a TLS stack.
We thereby also show, that the changes to get the approach implemented in a TLS stack were limited to typical steps when adding a new signature scheme.

Finally, we note that our scheme is not limited to managing delegations with CDNs.
Without any changes, it can be applied to scenarios where website providers serve content from multiple servers themselves.
Instead of distributing the private key to every server, delegating keys to the web servers from a central key server helps to significantly reduce the risk of leaking the secret key.

\section{Preliminaries}
In this section, we briefly recall some notions of identity-based signature schemes.

\subsection{Hierarchical Identity-Based Signatures}

(Hierarchical) Identity-based signature schemes~\cite{DBLP:conf/asiacrypt/GentryS02} extend the standard notion of signature schemes \ifdefined\full(cf. \Cref{app:sigs})\fi\ with an additional key delegation algorithm $\Del$ to support the delegation of signing rights based on identities. $\Verify$ takes the identities as an additional argument, and verification only succeeds if the key used during signing was delegated for those particular identities. We briefly recall the formal definition and adapt them for the multi-instance setting.
\begin{definition}
  A hierarchical identity-based signature scheme ($\HIBS$) consists of PPT algorithms $(\Setup,\Gen,\Del,\Sign,\Verify)$ such that:
  \begin{description}
    \item[$\Setup(1^\secpar,\ell)\colon$] On input of security parameter $\secpar$ and hierarchy parameter $\ell$, outputs public parameters $\pp$.
    \item[$\Gen(\pp)\colon$] On input of public parameters $\pp$, outputs a master signing key $\sk_\varepsilon$ and a verification key $\pk$ with message space $\Msg$.
    \item[$\Del(\sk_{\id'},\id)\colon$] On input of secret key \(\sk_{\id'}\) and \(\id\in\Id^{\leq\ell}\), outputs a secret key $\sk_{\id}$ for $\id$ iff \(\id'\) is a prefix of \(\id\), otherwise \(\sk_{\id'}\).
    \item[$\Sign(\sk_\id, \msg)\colon$] On input of a secret key $\sk_\id$ and a message $\msg \in \Msg$, outputs a signature $\sigma$.
    \item[$\Verify(\pk,\id,\msg,\sigma)\colon$] On input of a public key $\pk$, an identity $\id \in \Id^{\leq\ell}$, a message $\msg \in \Msg$ and a signature $\sigma$, outputs a bit $b$.
  \end{description}
\end{definition}
Below, we present the standard existential unforgeability under adaptively chosen message attacks ($\EUFCMA$ security) notion. It extends the standard notion of $\EUFCMA$ security for signature schemes by allowing the adversary to query keys for identities as long as they are not prefixes of the target identity.
\begin{experiment}[t]
  \begin{flushleft}
  The experiment has access to the following oracles:
  \begin{description}
    \item[$\Del'(\sk_\varepsilon,\id)\colon$] Stores $\id$ in $\mathcal{Q}^\Id$ and returns $\Del(\sk_\varepsilon, \id)$.
    \item[$\Sign'(\sk,\id,\msg)\colon$] This oracle computes $\sk_\id \gets \Del(\sk_\varepsilon, \id)$, $\sigma \gets \Sign(\sk_\id, \msg)$, adds $\msg$ to $\mathcal{Q}$, and returns $\sigma$.
  \end{description}
\end{flushleft}
  \textbf{Experiment} $\ExpSigEUFCMA{\HIBS,\advA}(\secpar,\ell)$
  \begin{algorithmic}
    \State $\pp \gets \Setup(1^\secpar,\ell)$, $(\pk, \sk) \gets \Gen(\pp)$
    \State $(\msg^*, \id^*, \sigma^*) \gets \advA^{\Del'(\sk_\varepsilon, \cdot),\Sign'(\sk_\varepsilon, \cdot, \cdot)}(\pk)$
    \State if $\Verify(\pk, \id^*,  \msg^*, \sigma^*) = 0$, return 0
    \State if $\msg^* \in \mathcal{Q}$, return 0
    \State if $\exists \id \in \mathcal{Q}^\Id$ such that $\id$ is a prefix of $\id^*$, return 0
    \State return 1
  \end{algorithmic}
  \caption{The $\EUFCMA$ experiment for a hierarchical identity-base signature scheme $\HIBS$.}
  \label{exp:hibs-eufcma}
\end{experiment}
\begin{definition}[$\EUFCMA$]
  For any PPT adversary $\advA$, we define the advantage in the $\EUFCMA$ experiment $\ExpSigEUFCMA{\HIBS,\advA}$ (cf. \Cref{exp:hibs-eufcma}) as
  \begin{align*}
    \AdvSigEUFCMA{\HIBS,\advA}(\secpar):=\Pr\left[\ExpSigEUFCMA{\HIBS,A}(\secpar) = 1\right] \text{.}
  \end{align*}
  A signature scheme $\Sigma$ is $\EUFCMA$-secure, if $\AdvSigEUFCMA{\Sigma,\advA}(\secpar)$ is a negligible function in $\secpar$ for all PPT adversaries $\advA$.
\end{definition}

\subsection{Naor Transform: Signatures from Identity-Based Encryption}

Based on transform first proposed by Naor~\cite{DBLP:conf/crypto/BonehF01}, hierarchical identity-based signature can be obtained from hierarchical identity-based encryption (\HIBE)~\cite{DBLP:conf/asiacrypt/GentryS02}, a generalization of identity-based encryption (\IBE)~\cite{DBLP:conf/crypto/BonehF01}. These encryption schemes enable delegation of decryption rights based on identities in a hierarchical fashion. Identities at some level can delegate secret keys to its descendant entities, but cannot decrypt ciphertexts intended for other (hierarchical) identities. Before we discuss the Naor transform, we recall a definition of {\HIBE}s supporting multiple instances (cf. \cite{DBLP:journals/iacr/DerlerRSS19}) below.
\begin{definition}[\HIBE]
  An hierarchical identity-based encryption (\HIBE) scheme with message space \(\Msg\) and identity space \(\Id^{\leq\ell}\), for some \(\ell\in\N\), consists of the PPT algorithms $(\Setup,\Gen,\Del,\Enc,\Dec)$:
\begin{description}
  \item[$\Setup(1^\secpar,\ell)\colon$] On input of security parameter $\secpar$ and hierarchy parameter $\ell$, outputs public parameters $\pp$.
  \item[$\Gen(\pp)\colon$] On input of public parameters $\pp$, outputs a keypair $(\pk,\sk_{\varepsilon})$.
  \item[$\Del(\sk_{\id'},\id)\colon$] On input secret key \(\sk_{\id'}\) and \(\id\in\Id^{\leq\ell}\), outputs a secret key $\sk_{\id}$ for $\id$ iff \(\id'\) is a prefix of \id, otherwise \(\sk_{\id'}\).
  \item[$\Enc(\pk,\id,\msg)\colon$] On input of a public key \pk, a message $\msg\in\Msg$ and an identity $\id\in\Id^{\leq\ell}$, outputs a ciphertext \(\ctxt_\id\) for \id.
  \item[$\Dec(\sk_{\id'},\ctxt_{\id})\colon$] On input of a secret key \(\sk_{\id'}\) and a ciphertext \(\ctxt_\id\), outputs $\msg\in\Msg\cup\{\bot\}$.
\end{description}
\end{definition}
For correctness, we require that all secret keys that are delegated for the identity (or a prefix) associated to a ciphertext, are able to decrypt it. More formally, we require for all \(\secpar,\ell\in\N\), all $\pp \gets \Setup(1^\secpar, \ell)$, all \((\pk,\sk_\varepsilon)\gets\Gen(\pp)\), all \(\msg\in\Msg\), all \(\id,\id'\in\Id^{\leq\ell}\) where \(\id'\) is a prefix of \id, all \(\sk_\id\gets\Del(\sk_{\id'},\id)\), all \(\ctxt_\id\gets\Enc(\pk,\id,\msg)\), \(\Dec(\sk_\id,\ctxt_\id)=\msg\) holds.

\begin{experiment}[t]
  \begin{flushleft}
  The experiment has access to the following oracles:
  \begin{description}
    \item[$\Del^1(\sk_\varepsilon,\id)\colon$] Stores $\id$ in $\mathcal{Q}$ and returns $\Del(\sk_\varepsilon, \id)$.
    \item[$\Del^2(\sk_\varepsilon,\id)\colon$] This oracle checks whether $\id$ is a prefix of $\id^*$ and returns $\bot$ if so. Otherwise it returns $\Del(\sk_\varepsilon, \id)$.
  \end{description}
\end{flushleft}
  \textbf{Experiment} $\ExpHIBEindcpa{\HIBE,\advA}(\secpar,\ell)$
  \begin{algorithmic}
    \State $\pp \gets\Setup(1^\secpar,\ell)$, $(\pk,\sk_\varepsilon) \gets \Gen(\pp)$, $b\getsr\{0,1\}$
    \State $(\id^*, \msg_0, \msg_1, \st) \gets \advA^{\Del^1(\sk_\varepsilon,\cdot)}(\pk)$
    \State if $\exists \id \in \mathcal{Q}$ such that $\id$ is a prefix of $\id^*$, return 0
    \State if $\msg_0, \msg_1 \notin \mathcal{M}$ or $|M_0|\neq |M_1|$, return 0
    \State $b^* \gets \advA^{\Del^2(\sk_\varepsilon,\cdot)}(\st,\Enc(\pk, \id_*, \msg_b))$
    \State if $b = b^*$ return then $1$, else return $0$
  \end{algorithmic}
  \caption{The $\HIBEINDCPA$ experiment for a $\HIBE$ scheme.}
  \label{exp:hibe-ind-cpa}
\end{experiment}

Next we recall the standard security notion of indistinguishability under chosen ciphertext attacks. Here an adversary chooses a target identity and two messages and may query keys for identities as long as the queries identities are not a prefix of the target identity. Then, given the encryption of one of the two messages under the target identity, the adversary needs to determine which message was encrypted. The adversary should not be able to win this experiment better than guessing.
\begin{definition}
  For any PPT adversary \advA, we define the advantage in the $\HIBEINDCPA$ experiment $\ExpHIBEindcpa{\HIBE,\advA}$ (cf. \Cref{exp:hibe-ind-cpa}) as
  \begin{align*}
    \AdvHIBEindcpa{\HIBE,\advA}(\secpar,\ell):=\left|\Pr\left[\ExpHIBEindcpa{\HIBE,A}(\secpar,\ell) = 1\right] - \frac{1}{2}\right| \text{,}
  \end{align*}
  for an integer $\ell\in\N$. A $\HIBE$ is $\HIBEINDCPA$-secure, if $\AdvHIBEindcpa{\HIBE,\advA}\allowbreak (\secpar,\ell)$ is a negligible function in $\secpar$ for all PPT adversaries $\adv$.
\end{definition}

Subsequently, we provide a concrete HIBE construction on bilinear groups and, in particular, present a variant of the Boneh-Boyen-Goh (BBG) HIBE~\cite{DBLP:conf/eurocrypt/BonehBG05} with an explicit $\Setup$ algorithm to generate shared parameters in \Cref{sm:bbg_cca}. We choose to instantiate our approach using asymmetric bilinear groups in the type 3 setting as they represent the state-of-the-art regarding efficiency and similarity of the security levels of the base and target groups. Let $\BilGen$ be an algorithm that, on input a security parameter $1^\secpar$, outputs $\mathsf{BG}=(p, e, \G_1, \G_2, \G_T, g, \hat g) \gets \BilGen(1^\secpar)$, where $\G_1$, $\G_2$, $\G_T$ are groups of prime order $p$ with bilinear map $e\colon \G_1 \times \G_2 \to \G_T$ and generators $g, \hat{g}$ of $\G_1$ and $\G_2$, respectively.
\begin{scheme}[t]
\begin{description}
  \item[$\underline{\Setup(1^\secpar,\ell)}\colon$] Generate a bilinear group $(p, e, \G_1, \G_2, \G_T, g, \hat g) \gets \BilGen(1^\secpar)$ and $g_2,g_3,h_1,\ldots,h_{\ell} \getsr \G_1$, fix a hash function $H\colon\{0,1\}^*\rightarrow \Z_p^*$ and return $\pp\gets (H, p, e, \G_1, \G_2, \G_T, g, \hat g, g_2,g_3,h_1,\ldots,h_{\ell})$.
  \item[$\underline{\Gen({\pp})}\colon$] Choose $\alpha\getsr \Z_p$ and return $(\pk, \sk)\gets (\hat g^\alpha, g_2^\alpha)$.
  \item[$\underline{\Del({\sk_{\id},\id})}\colon$] Parse $\id$ as $(I_1,\ldots,I_k)$ with $k\leq\ell$,
    \begin{compactitem}
      \item[-] If $\sk_{\id}$ is the master secret $g_2^\alpha$, sample $v\getsr \Z_p$ and return $(g_2^\alpha\cdot (h_1^{H(I_1)} \cdots h_k^{H(I_k)} \cdot g_3)^v, \hat g^v, h_{k+1}^v,\ldots,h_{\ell}^v)$,
      \item[-] else assume that $\sk_{\id}$ is $(g_2^\alpha\cdot (h_1^{H(I_1)} \cdots h_{k-1}^{H(I_{k-1})} \cdot g_3)^{v'}, \hat g^{v'}, h_{k}^{v'},\ldots,h_{\ell}^{v'})=(a_0,a_1,b_{k},\ldots,b_{\ell})$, sample $w\getsr \Z_p$ and output $(a_0\cdot b_k^{H(I_{k})} \cdot (h_1^{H(I_1)} \cdots h_{k}^{H(I_{k})} \cdot g_3)^{w},a_1\cdot \hat g^w, b_{k+1}\cdot h_{k+1}^w,\ldots,b_{\ell}\cdot h_{\ell}^w)$.
    \end{compactitem}
  \item[$\underline{\Enc(\pk,\id, \msg)}\colon$] Parse $\id$ as $(I_1,\ldots,I_k) \in (\Z_p^*)^{k}$ with $k \leq \ell$, sample $s\getsr \Z_p$, and return $(C_1,C_2,C_3)\gets (e(g_2,\pk)^s\cdot \msg, \hat g^s, (h_1^{H(I_1)}\cdot \ldots \cdot h_{k}^{H(I_k)} \cdot g_3)^s)$.
  \item[$\underline{\Dec(\sk_{\id},\ctxt_{\id})}\colon$] Consider $\id=(I_1,\ldots,I_k)$ with $k\leq\ell$, parse $\sk_{\id}$ as $(a_0,a_1,b_{k+1},\ldots,b_\ell)$, $\ctxt_{\id}$ as $(C_1,C_2,C_3)$. Return $\msg \gets C_1 \cdot e(C_3, a_1)\cdot e(a_0,C_2)^{-1}$.
\end{description}
  \caption{$\HIBEINDCPA$-secure version of the BBG \HIBE.}
\label{sm:bbg_cca}
\end{scheme}

We are now ready to recall the IBE-to-signature transformation from Naor~\cite{DBLP:conf/crypto/BonehF01}: a signature on a message $\msg$ under $\sk_{\id}$ is just a $\HIBE$ secret key delegated for $\id' = (\id \| \msg)$ which is verified by encrypting a random message for $\id'$ and then trying to decrypt it using the key in the signature. More generally, it converts any level $\ell+1$ $\HIBE$ to a level $\ell$ $\HIBS$ (cf. \cite{DBLP:series/ciss/KiltzN09}) as depicted in \Cref{sm:naor}. In the special case with $\ell=0$, we obtain a signature scheme from an $\IBE$. Prominent signatures schemes that are based on this technique includes BLS~\cite{DBLP:conf/asiacrypt/BonehLS01}.

\begin{scheme}[ht]
  \begin{description}
    \item[$\underline{\Setup(1^\secpar, \ell)}\colon$] Let $\pp \gets \HIBE.\Setup(1^\secpar,\ell+1)$ and return $\pp$.
    \item[$\underline{\Gen({\pp})}\colon$] Run $(\pk, \sk) \gets \HIBE.\Gen(\pp)$ and return $(\pk, \sk)$.
    \item[$\underline{\Del(\sk_{\id'}, \id)}\colon$] Return $\sk_\id \gets \HIBE.\Del(\sk_{\id'}, \id)$.
    \item[$\underline{\Sign(\sk_\id, \msg)}\colon$] Parse $\id$ as $\id_1, \ldots, \id_l$ and return $\sigma \gets \HIBE.\Del(\sk, (\id_1, \ldots, \id_\ell, \msg))$.
    \item[$\underline{\Verify(\pk, (\id_i)_{i=1}^\ell, \msg,\sigma)}\colon$] Choose $\msg' \getsr \HIBE.\Msg$ and compute $\ctxt \gets \HIBE.\Enc(\pk, (\id_1, \ldots, \id_\ell, \msg), \msg')$. Return $1$ if $\HIBE.\Dec(\sigma, \ctxt) = \msg'$, otherwise return $0$.
  \end{description}
  \caption{(Hierarchical identity-based) Signature scheme obtained by applying Naor-transform to $\HIBE$.}
  \label{sm:naor}
\end{scheme}

\begin{theorem}[\cite{DBLP:conf/provsec/CuiFHIZ07}]
  If $\HIBE$ is $\HIBEINDCPA$ secure and has a message space that is exponentially large in the security parameter, then \Cref{sm:naor} is $\EUFCMA$ secure.
\end{theorem}

\section{$\TBIDS$: Time-bound Identity-Based Signatures}
\label{sec:crypto}

In this section, we introduce a special-case of hierarchical identity-based signatures: time-bound identity-based signatures. The scheme keeps the identity-based delegation mechanism provided by identity-based signatures, but at the same time it is modified to support an epoch besides an identity. The epoch is designed to be used in a way that the identity-based delegations are bound to a time period. Thereby, even if one has received a delegated key, this key is not useful before or after the corresponding period. Additionally, we introduce the notion of forward-security for the master secret key that is used for all delegations at almost no additional cost.

\subsection{Syntax and Definitions}
\label{sec:cryptoDef}

First we start with the syntax of time-bound identity-based signatures. Therefore, we extend the definition of identity-based signatures by adding an epoch to the sign and verification algorithms as well as the delegation algorithm.
\begin{definition}[Time-bound Identity-Based Signatures] \label{def:tbids}
  A time-bound identity-based signatures scheme $\TBIDS$ with identity-space $\Id$ consists of the PPT algorithms $(\Gen, \Sign, \Verify, \Del)$, which are defined as follows:
  \begin{description}
    \item[$\Gen(1^\secpar, n)\colon$] On input of a security parameter $\secpar$ and maximal number of epochs $n$, outputs a signing key $\sk$ and a verification key $\pk$ with associated message space $\Msg$.\footnote{We allow $n = \infty$ to denote an unbounded number of epochs.}
    \item[$\Sign(\sk_{i,\id}, \msg)\colon$] On input of a secret key $\sk_{i,\id}$ for an identity $\id \in \Id$ and an epoch $i \in [n]$ and a message $\msg \in \Msg$, outputs a signature $\sigma$.
    \item[$\Verify(\pk,i,\id, \msg,\sigma)\colon$] On input of a public key $\pk$, an identity $\id \in \Id$, an epoch $i \in [n]$, a message $\msg \in \Msg$ and a signature $\sigma$, outputs a bit $b \in \{0,1\}$.
    \item[$\Del(\sk, i,\id)\colon$] On input of a secret key $\sk$, an identity $\id \in \Id$ and an epoch $i \in [n]$, outputs a secret key $\sk_{i,\id}$.
  \end{description}
\end{definition}

Such a scheme is considered correct if for all security parameters $\secpar \in \N$ and $n = n(\secpar) \in \N$, for all $(\sk, \pk) \gets \Gen(1^\secpar, n)$, for all $\id \in \Id$, for all $i \in [n]$, for all extracted keys $\sk_{i,\id} \gets \Del(\sk, i,\id)$, for all $\msg \in \Msg$, we have that
\[
  \Pr\mleft[\Verify(\pk, i,\id, \msg, \Sign(\sk_{i,\id}, \msg)) = 1\mright] = 1\text{.}
\]

For unforgeability, the standard security notion of $\EUFCMA$ security is extended to cover the epoch for the delegation of keys. The goal of the adapted notion is to have the adversary select a target epoch and identity. The scheme is then considered unforgeable as long as the adversary is unable to forge a signature for the selected target epoch and identity -- even if the adversary has access to keys for the target identity for other epochs or keys for other identities.
\begin{experiment}[ht]
  \begin{flushleft}
  The experiment has access to the following oracles:
  \begin{description}
    \item[$\Sign'(\sk,i,\id, \msg)\colon$] This oracle computes $\sk_{i,\id} \gets \Del(\sk, i, \id)$, $\sigma \gets \Sign(\sk_{\id,i}, \msg)$, adds $\msg$ to $\mathcal{Q}$, and returns $\sigma$.
    \item[$\Del'(\sk, i, \id)\colon$] Stores $(i, \id)$ in $\mathcal{Q}^\Id$ and returns $\Del(\sk, i, \id)$.
  \end{description}
\end{flushleft}
  \textbf{Experiment} $\ExpSigEUFCMA{\TBIDS,\advA}(\secpar, n)$
  \begin{algorithmic}
    \State $\pp \gets \Setup(1^\secpar,n)$, $(\pk, \sk) \gets \Gen(\pp)$
    \State $(\msg^*, i^*, \id^*, \sigma^*) \gets \advA^{\Del'(\sk_, \cdot),\Sign'(\sk, \cdot, \cdot, \cdot)}(\pk)$
    \State if $\Verify(\pk, i^*, \id^*,  \msg^*, \sigma^*) = 0$, return $0$
    \State if $\msg^* \in \mathcal{Q}$, return $0$
    \State if $(i^*, \id^*) \in \mathcal{Q}^\Id$, return $0$
    \State return $1$
  \end{algorithmic}
  \caption{The $\EUFCMA$ experiment for a time-bound identity-based signature scheme $\TBIDS$.}
  \label{exp:tbids-eufcma}
\end{experiment}
\begin{definition}[$\EUFCMA$]
  For any PPT adversary $\advA$, we define the advantage in the $\EUFCMA$ experiment $\ExpSigEUFCMA{\TBIDS,\advA}$ (cf. \Cref{exp:tbids-eufcma}) as
  \begin{align*}
    \AdvSigEUFCMA{\TBIDS,\advA}(\secpar,n):=\Pr\left[\ExpSigEUFCMA{\TBIDS,\advA}(\secpar,n) = 1\right] \text{,}
  \end{align*}
  for an integer $n\in\N$. A time-bound identity-based signature scheme $\TBIDS$ is $\EUFCMA$-secure, if $\AdvSigEUFCMA{\TBIDS,\advA}(\secpar,n)$ is a negligible function in $\secpar$ for all PPT adversaries $\advA$.
\end{definition}

\begin{remark}
We note that while the syntax is very similar to forward-secure identity-based signatures~\cite{DBLP:conf/asiacrypt/GalindoHK06}, the goals are different. In such a signature scheme, the forward-security of keys is considered per identity, i.e., each user gets a master secret key for their identity. Hence, users are able to update their own keys to the next epoch. For $\TBIDS$, the goal is to prevent any updates, and thus updates of delegatable keys are not possible.
\end{remark}
Henceforth, we also introduce a forward-secure version of time-bound identity-based signatures:
\begin{definition} %
  A forward-secure time-bound identity-based signatures scheme $\fsTBIDS$ extends \Cref{def:tbids} with a PPT algorithm $\Update$ such that:
  \begin{description}
    \item[$\Update(\sk_{i-1})\colon$] On input of a secret key $\sk_i$ for epoch $i-1$, outputs a secret key $\sk_i$ for epoch $i$.
    \item[$\Del(\sk_i,\id)\colon$] On input of a secret key $\sk_i$ for epoch $i$ and an identity $\id \in \Id$, outputs a secret key $\sk_{i,\id}$.
\end{description}
\end{definition}
Correctness is defined as before, but the derived keys are obtained via $\sk_i \gets \Update(\sk_{i-1})$ and then $\sk_{i,\id} \gets \Del(\sk_i, \id)$. Subsequently, we present unforgeability for $\fsTBIDS$. In this case, the adversary may also query the secret key for an epoch by querying the $\Update'$ oracle. The adversary then has to forge a signature for an epoch before any epoch queried via $\Update'$ and is also not allowed to query delegated keys for the target epoch and identity.
\begin{experiment}[ht]
  \begin{flushleft}
  The experiment has access to the following oracles:
  \begin{description}
    \item[$\Sign'(\sk_0, i, \id, \msg)\colon$] This oracle computes $\sk_{i,\id} \gets \Del(\sk_0, i, \id)$, $\sigma \gets \Sign(\sk_{\id,i}, \msg)$, adds $\msg$ to $\mathcal{Q}$, and returns $\sigma$.
    \item[$\Del'(\sk_0, i, \id)\colon$] Stores $(i, \id)$ in $\mathcal{Q}^\Id$ and returns $\Del(\sk_0, i, \id)$.
    \item[$\Update'(\sk_0, i)\colon$] Stores $i$ in $\mathcal{Q}^t$, runs $\sk_{j} \gets \Update(\sk_{j-1})$ for $j \in [i]$, and returns $\sk_i$
  \end{description}
\end{flushleft}
  \textbf{Experiment} $\ExpSigfsEUFCMA{\fsTBIDS,\advA}(\secpar, n)$
  \begin{algorithmic}
    \State $\pp \gets \Setup(1^\secpar,n)$, $(\pk, \sk) \gets \Gen(\pp)$
    \State $(\msg^*, i^*, \id^*, \sigma^*) \gets \advA^{\Del'(\sk_, \cdot),\Sign'(\sk, \cdot, \cdot, \cdot),\Update'(\sk_0, \cdot)}(\pk)$
    \State if $\Verify(\pk, i^* - 1, \id^*,  \msg^*, \sigma^*) = 0$, return 0
    \State if $\msg^* \in \mathcal{Q}$, return 0
    \State if $(i^* - 1, \id^*) \in \mathcal{Q}^\Id$, return 0
    \State if $\{1,\ldots,i^* - 1\} \cap \mathcal{Q}^t \neq \emptyset$, return 0
    \State return 1
  \end{algorithmic}
  \caption{The $\EUFCMA$ experiment for a forward-secure time-bound identity-based signature scheme $\TBIDS$.}
  \label{exp:fs-tbids-eufcma}
\end{experiment}
\begin{definition}[$\fsEUFCMA$]
  For any PPT adversary $\advA$, we define the advantage in the $\fsEUFCMA$ experiment $\ExpSigfsEUFCMA{\TBIDS,\advA}$ (cf. \Cref{exp:fs-tbids-eufcma}) as
  \begin{align*}
    \AdvSigEUFCMA{\TBIDS,\advA}(\secpar,n):=\Pr\left[\ExpSigEUFCMA{\fsTBIDS,\advA,n}(\secpar) = 1\right] \text{,}
  \end{align*}
  for an integer $n\in\N$. A time-bound identity-based signature scheme $\fsTBIDS$ is $\fsEUFCMA$-secure, if $\AdvSigfsEUFCMA{\fsTBIDS,\advA}(\secpar, n)$ is a negligible function in $\secpar$ for all PPT adversaries $\advA$.
\end{definition}

Before discussing constructions of the scheme, we want to note that using generic transformation from \cite{DBLP:conf/asiacrypt/GalindoHK06}, $\TBIDS$ and $\fsTBIDS$ can be turned into (forward-secure) identity-based proxy signatures. Indeed, if the identity for delegation is the delegatee's public key, and each signature is extended with a signature by that key, we obtain a proxy signature scheme.

\subsection{Generic \TBIDS\ Construction}
\label{sec:generic-tbids}

First, we start with a generic construction of \TBIDS\ from identity-based signatures. Observe that to achieve an \EUFCMA-secure \TBIDS\ scheme, delegation of the keys is relative to both epochs and identities. So, the idea is to map \TBIDS's epoch and identity to an identity of the underlying identity-based signature scheme. Without the goal to achieve forward-secrecy, this mapping is simply a concatenation of epoch and identity. As the epoch is only used to partition time, it can be viewed as giving a specific meaning to a part of the identity of an \IBS.  We present the scheme in \Cref{sm:tbids-ibs}.

\begin{scheme}[ht]
\begin{description}
  \item[$\underline{\Setup(1^\secpar,n)}\colon$] Let $\pp \gets \IBS.\Setup(1^\secpar)$ with $\IBS.\Id = \{0,1\}^{\lceil \log_2(n) \rceil} \times \Id$ and return $\pp$.
  \item[$\underline{\Gen(\pp)}\colon$] Return $(\pk, \sk) \gets \IBS.\Gen(\pp)$.
  \item[$\underline{\Del(\sk,i,\id)}\colon$] Return $\sk_{i,\id} \gets \IBS.\Del(sk, i \| \id)$.
  \item[$\underline{\Sign(\sk_{i,\id}, \msg)}\colon$] Return $\sigma \gets \IBS.\Sign(\sk_{i,\id}, \msg)$.
  \item[$\underline{\Verify(\pk, i, \id, \msg, \sigma)}\colon$] Return $\IBS.\Verify(\pk, i \| \id, \msg, \sigma)$.
\end{description}
\caption{Generic $\TBIDS$ scheme from any $\IBS$.}
\label{sm:tbids-ibs}
\end{scheme}

$\EUFCMA$ security follows directly the underlying $\IBS$ scheme:
\begin{theorem}
  \label{thm:tbids-ibs}
  If $\IBS$ is $\EUFCMA$-secure, then \Cref{sm:tbids-ibs} is $\EUFCMA$-secure, i.e. if there is an $\EUFCMA$-adversary $\advA$ against $\TBIDS$, then there exists an $\EUFCMA$-adversary against $\IBS$ with
  \[
    \AdvSigEUFCMA{\TBIDS,\advA}(1^\secpar,n) = \AdvSigEUFCMA{\IBS,\advB}(1^\secpar)
    \text{.}
  \]
\end{theorem}
\begin{proof}
  Let $\advA$ be an $\EUFCMA$-adversary against $\TBIDS$. We build an $\EUFCMA$-adversary $\advB$ against $\IBS$ with $\IBS.\Id = \{0,1\}^k \times \Id$. Indeed, if $\advB$ is started on a public key $\pk_\IBS$, we run $\advA$ with $\pk_{\TBIDS} \gets \pk_\IBS$ and the oracles are simulated honestly using the corresponding $\IBS$ oracles.
  Once $\advA$ outputs a forgery $\sigma^*$ for $i^*$, $\id^*$, $\msg^*$ and forwards it as a $\IBS$ forgery for $i^* \| \id^*$ and $\msg^*$. A valid forgery for $\TBIDS$ is also a valid forgery for $\IBS$, since if $(i^*, \id^*)$ was not queried on the $\TBIDS.\Del'$ oracle, then neither was $i^* \| id^*$ queried on $\IBS.\Del'$. Similarly, the same is true for $m^*$.
\end{proof}

\subsection{Forward-Secure Construction from \HIBS}
\label{sec:fstbids}

We now consider a direct construction of \fsTBIDS\ from \HIBS. In contrast to \Cref{sec:generic-tbids}, the main difference is the management of the master secret key. The key has to be updated from one epoch to another while still allowing identity-based derivations. Therefore, we follow a standard approach due to Canetti, Halevi and Katz~\cite{DBLP:conf/eurocrypt/CanettiHK03} to achieve forward-secrecy. The idea is to manage the epochs in a tree such that the individual epochs are the leaves of the tree, whereas epoch $0$ is placed at the root. Consequently, when mapping the epochs to keys, we place the master secret key, $\sk_0$, at the root. Once the key is updated to the first epoch, it should no longer be possible to access $\sk_0$, yet deriving keys for all other epochs. Hence, we remove the root key and keep the keys delegated for the two child nodes. For the left child node this process is repeated until the left most leaf node is reached. Later, when updating to the next epoch, one removes the key for the current epoch and, starting from the associated leaf node searches for the first non-removed sibling along to the path to the root node. If this node is not a leaf node, then as in the first step, the associated key is removed and those for the child nodes are kept. For the left child the process is again iterated until a leaf node is reached.

More formally, the key update algorithm traverses the tree in a depth-first manner. Hence the keys are viewed as stack, and when visiting a node, the derived secret keys are pushed onto the stack. We define an algorithm \DFEval\ that performs the stack manipulation (adapted from \cite{DBLP:conf/pkc/DerlerKLRSS18}). Notation-wise, we denote the root node with $w^\varepsilon$ and all other nodes are encoded as binary strings, i.e., for a node $w^{(i)}$, we denote child nodes as $w^{(i0)}$ and $w^{(i1)}$. For encoding the epoch as identities, we define $\binid(i)$ which maps each bit of $i$ (considered as $\lceil \log_2 n \rceil$-bit number) to $0$ and $1$ identities, i.e. $\binid(b^{(1)}\ldots b^{(n)}_2) = (b^{(1)}, \ldots, b^{(n)})$.
\begin{description}
  \item[$\DFEval(s)\colon$] On input of the stack $s$, perform the following steps:
    \begin{itemize}
      \item Pop the topmost element, $(\sk^{(i)})$, from the stack $s$.
      \item Repeat while $w^i$ is an internal node:
        \begin{itemize}
          \item Set $\sk^{(i0)} \gets \HIBS.\Del(\sk^{(i)}, \binid(i0))$
          \item Set $\sk^{(i1)} \gets \HIBS.\Del(\sk^{(i)}, \binid(i1))$
          \item Push $\sk^{(i1)}$ onto $s$ and set $i \gets i0$
        \end{itemize}
      \item Return $\sk^{(i)}$ and the new stack $s$
    \end{itemize}
\end{description}
We give the full scheme in \Cref{sm:tbids-hibe}.

\begin{scheme}[t]
\begin{description}
  \item[$\underline{\Setup(1^\secpar,n)}\colon$] Let $\pp \gets \HIBS.\Setup(1^\secpar, 1 + \lceil \log_2(n) \rceil)$ with $\HIBS.\Id$ such that $\{0,1\} \subset \HIBS.\Id$ and $\Id \subset \HIBS.\Id$. Return $\pp$
  \item[$\underline{\Gen(\pp)}\colon$] Set $(\pk, \sk_0) \gets \HIBS.\Gen(\pp)$, push $\sk_0$ onto the empty stack $s$, and return $(\pk, (\sk_0, s))$.
  \item[$\underline{\Update(\sk_i)}\colon$] Parse $\sk_i$ as $(\sk_i, s)$ and return $\DFEval(\sk^{(i)}, s)$.
  \item[$\underline{\Del(\sk_i,\id)}\colon$] Parse $\sk_i$ as $(\sk^{(i)}, \cdot)$ and return $\sk_{i,\id} \gets \HIBS.\Del(\sk^{(i)}, (\binid(i), \id))$
  \item[$\underline{\Sign(\sk_{i,\id}, \msg)}\colon$] Return $\HIBS.\Sign(\sk_{i,\id}, (\binid(i), \id), \msg)$.
  \item[$\underline{\Verify(\pk, i, \id, \msg, \sigma)}\colon$] Return $\HIBS.\Verify(\pk, (\binid(i), \id), \msg)$.
\end{description}
\caption{Generic $\fsTBIDS$ scheme from a $\HIBS$.}
\label{sm:tbids-hibe}
\end{scheme}

Next we show that \Cref{sm:tbids-hibe} is in fact $\fsEUFCMA$-secure.
\begin{theorem}
  \label{thm:fstibds}
  If $\HIBS$ is $\EUFCMA$-secure, then \Cref{sm:tbids-hibe} is \\$\fsEUFCMA$-secure. For an $\EUFCMA$-adversary $\advA$ against \Cref{sm:tbids-hibe}, there exists an $\EUFCMA$-adversary against $\advB$ such that
  \[
    \AdvSigfsEUFCMA{\fsTBIDS,\advA}(1^\secpar,n) = \AdvSigEUFCMA{\HIBS,\advB}(1^\secpar, \lceil \log_2(n) \rceil + 1)
    \text{.}
  \]
\end{theorem}
\begin{proof}
  The proof follows the same idea as the proof of \Cref{thm:tbids-ibs}. We note, that the $\Update'$ oracle cannot be simulated honestly. However, it is possible to simulate it in the following way: given an epoch $i$, directly obtain $\sk_i$ as $\HIBS.\Del'(\sk, \binid(i))$ and build the stack by querying $\HIBS.\Del'$ by querying the keys for all right sibling nodes along the path from the $i$-th leave to the root.

  Note that for a valid forgery, we have that for no $i < i^*$ the $\Update'$ oracle was queried. Consequently, the $\Del'$ oracle of the $\HIBS$ is never queried on a prefix of $\binid(i^*)$, and thus the forgery is also valid for $\HIBS$.
\end{proof}

As for $\TBIDS$, $\EUFCMA$-security of the underlying $\HIBS$ implies the $\EUFCMA$-security of \Cref{sm:tbids-hibe}:
\begin{theorem}
  If $\HIBS$ is $\EUFCMA$-secure, then \Cref{sm:tbids-hibe} is $\EUFCMA$-secure. For an $\EUFCMA$-adversary $\advA$ against \Cref{sm:tbids-hibe}, there exists an $\EUFCMA$-adversary against $\advB$ such that
  \[
    \AdvSigEUFCMA{\fsTBIDS,\advA}(1^\secpar,n) = \AdvSigEUFCMA{\HIBS,\advB}(1^\secpar, \lceil \log_2(n) \rceil + 1)
    \text{.}
  \]
\end{theorem}
The proof is the same as for \Cref{thm:tbids-ibs} with the simulation of the $\Update$-oracle as in \Cref{thm:fstibds}; hence we do not repeat it.

\subsection{Practical Considerations}
\label{sec:practical-precomputation}

When implementing the $\TBIDS$ and $\fsTBIDS$ schemes from \Cref{sm:tbids-ibs} and \Cref{sm:tbids-hibe}, respectively, based on the BBG HIBE from \Cref{sm:bbg_cca}, some optimizations for efficient signing and verification can be applied. We discuss them below.

\paragraph{Deterministic Verification}
Contrary to verification algorithms of standard signature schemes (e.g., EdDSA and ECDSA), the verification algorithm of Naor-transformed signature schemes is probabilistic. Consequently, the verifier needs access to a good source of randomness which makes implementations even more complex. Taking a step back, the goal of the test encryption and decryption is to decide whether the derived key in the signature is a correctly derived key for the identity. Depending on the underlying \HIBE\ it might, however, be possible to check the correctness of the key with a deterministic method~\cite{DBLP:conf/asiacrypt/BonehLS01,DBLP:conf/eurocrypt/Waters05}.

Theoretically, one could attach a non-interactive zero-knowledge proof that certifies the correctness of the key derivation. Instead of the test decryption, the verification of signatures would then verify the zero-knowledge proof. More practically, a \HIBE\ like BBG allows one to check the derivation of the keys using pairing equations. Assuming a derived key for the last level, $\sk_{\id_1, \ldots, \id_\ell} = (\sk_1, \sk_2)$ we have that
\[
  e\mleft(h_1^{H(\id_1)} \cdots h_\ell^{H(\id_\ell)} \cdot g_3, \sk_2\mright) \cdot e\mleft(g_2, \pk\mright) = e\mleft(\sk_1, \hat{g}\mright) \text{.}
\]
Conversely, if we have a key that satisfies the equation, then it is also able to decrypt any ciphertext. Thereby, we obtain a faster and deterministic verification algorithm.

\paragraph{Faster Signing with Precomputation}
Given that the IDs $\id_1$ to $\id_{\ell-1}$ are fixed for the signer, signing can be implemented more efficiently without computing $h_i^{H(\id_i)}$ all the time. Given the secret key $\sk_{i,\id} = (a_0, a_1, b_\ell)$, the signature can be computed as
\[
  \mleft(a_0 \cdot b_\ell^{H(\msg)} \cdot \mleft(t \cdot h_k^{H(\msg)}\mright)^w, a_1 \cdot \hat{g}^w\mright)
\]
where $t \gets h_1^{H(\id_1)} \cdots h_{\ell-1}^{H(\id_2)} \cdot g_3$. This $t$ can be either computed when deserializing the secret key or directly be stored in the serialized key. It can be reused for all subsequent signing operations, thus saving $O(\ell)$ group operations and rendering signing runtime complexity independent of $\ell$. Similarly, when verifying signatures this value can also be precomputed when deserializing the public key. In this case, there is only a benefit if multiple signatures from the same signer are verified at the same time.

\section{Integration into TLS}
\label{sec:integration}
We integrated our $\TBIDS$ scheme into TLS 1.3, the newest standard of the TLS family.
Since our implementation does not require any changes to the TLS handshake or other parts of the protocol, $\TBIDS$ could fit older versions of TLS as well.

\paragraph{Status quo: Sharing the Secret Key}
The typical CDN setup involves a client, a CDN, and an origin server, as shown in \Cref{fig:tls:setup}.
A client, usually a web browser, initiates the connection and wants to retrieve some remote server data.
To offload work and improve the connection speed, the origin server delegates some or all of its work to a CDN.
That happens either by uploading data to the CDN beforehand or using the CDN as a cache between the client and the origin server.
In the latter case, a CDN node retrieves data from the origin server as soon as a client requests it and keeps it locally, making it faster available for other clients in the area.
Therefore, the origin server's DNS entry needs to point to the local CDN node to ensure the requests are forwarded to the CDN.
The CDN can then serve data on behalf of the origin server, and consequently, acts in the origin server's name.
All connections between client and CDN, and between CDN and origin server are secured by TLS and therefore authenticated.

\begin{figure}[th]
    \centering
    \newcommand{\domain}{\textcolor{black}{\ensuremath{\mathsf{ccsw.io}}}}

\begin{tikzpicture}
  \begin{umlseqdiag}
  
    \umlobject[no ddots, x=7]{Server} 
    \umlobject[no ddots, x=0]{CA} 
    
    \begin{umlcallself}
      [side=left,
      padding=1, 
      op={generate \makecell[l]{\skServer\\$pk_{server}$}}]
      {Server} 
  
    \begin{umlcallself}
      [side=left,
      padding=1,
      dt=5, 
      op={generate $CSR$ \makecell[l]{pk: $pk_{server}$\\CN: \domain}}]
      {Server} 
    \end{umlcallself} 
    
    \begin{umlcall}
      [dt=4,
      op={$CSR$}, 
      return={$cert_{server}$}]  %
      {Server}{CA} 
      
      \begin{umlcallself}
      [padding=1, 
      dt=2, 
      op={$\Sign$ CSR}]{CA} %
      \end{umlcallself} 
    \end{umlcall} 
    
    \umlcreatecall[dt=5, no ddots, x=0]{Server}{CDN}
    
    \begin{umlcall}
      [padding=0,
      dt=2,
      op=\skServer]{Server}{CDN} 
      \begin{umlcall}
        [padding=0,
        op={$cert_{server}$}]{Server}{CDN} 
      \end{umlcall} 
    \end{umlcall} 
        
    \end{umlcallself} 
    
  \end{umlseqdiag} 
\end{tikzpicture}
    \caption{Typical setup of a PKI involving a CDN. The origin server shares its secret key \skServer{} with the CDN.}
    \label{fig:tls:setup}
\end{figure}
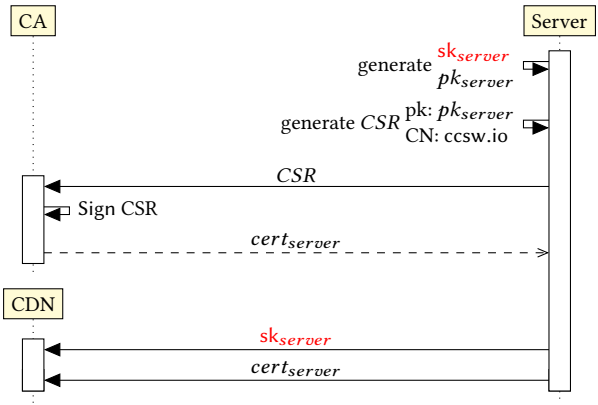

\begin{figure}[th]
    \centering
                          
                       
\begin{tikzpicture} 
  \begin{umlseqdiag} 
  
    \umlobject[no ddots]{Client} 
    \umlobject[no ddots, x=7]{CDN}

    \begin{umlcall}
      [padding=0,
      op={ClientHello}]{Client}{CDN} %

    \begin{umlcall}
      [padding=1,
      type=return, 
      op={ServerHello}]{CDN}{Client} 
    \end{umlcall} 
    
    \begin{umlcall}
      [padding=1, 
      dt=2,
      type=return, op={Encrypted Extensions}]{CDN}{Client} 
    \end{umlcall} 
    
    \begin{umlcall}
      [padding=1, 
      dt=2,
      type=return, 
      op={Server Certificate ($cert_{server}$)}]{CDN}{Client} 
    \end{umlcall}

    \begin{umlcallself}
      [padding=1, 
      dt=2,
      side=left,
      op={\makecell[l]{$cv=\Sign(\skServer, handshake)$}}]{CDN} %
    \end{umlcallself}
    
    \begin{umlcall}
      [padding=1, 
      dt=4,
      type=return, 
      op={Certificate Verify ($cv$)}]{CDN}{Client} 
    \end{umlcall} 
    
    \begin{umlcall}
      [padding=1, 
      dt=2,
      type=return, 
      op={Handshake Finished}]{CDN}{Client} 
    \end{umlcall} 
    
    \begin{umlcallself}
      [padding=1, 
      dt=2,
      op={$\Verify(pk_{server}, handshake, cv)$}]{Client} 
    \end{umlcallself} 
    
    \begin{umlcallself}
      [padding=1, 
      dt=2,
      op={Verify PKI \& cert validity}]{Client} 
    \end{umlcallself} 
    
    \begin{umlcall}
      [padding=0, op={Client Finished}]{Client}{CDN} 
    \end{umlcall} 
    \end{umlcall}

  \end{umlseqdiag} 
\end{tikzpicture}
    \caption{TLS flow of the communication of a client with a CDN. The CDN uses the origin server secret key \skServer{} to sign the TLS handshake.}
    \label{fig:tls:handshake}
\end{figure}
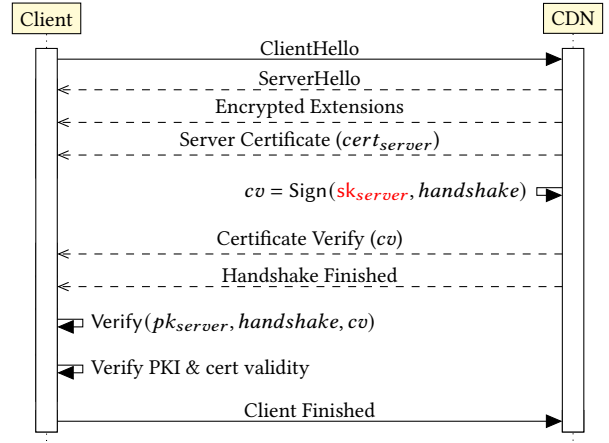

To serve content securely using TLS, a CDN needs to be able to sign TLS sessions in the name of the origin server, as shown in \Cref{fig:tls:handshake}.
Usually, that means a CDN needs access to the secret key \skServer{} of a certificate, which is valid for the origin server domain.
One way to achieve that is to share the secret key with the CDN, as shown in \Cref{fig:tls:setup}.
Another approach is to allow the CDN to acquire a valid certificate on its own.

Outsourcing a secret key is often neither easily possible nor desirable.
For example, in some industries, there might be regulations in place which disallow that.
Also, the idea of a secret key is for it to remain secret.
Sharing a secret key, or even allowing a CDN to acquire its own, implicates loss of control, e.g., in case of a CDN compromise or malicious CDN.

\begin{figure}[th]
    \centering
    \newcommand{\domain}{\textcolor{black}{\ensuremath{\mathsf{ccsw.io}}}}
\newcommand{\epochform}{\textcolor{black}{\ensuremath{\mathsf{epoch}}}}

\begin{tikzpicture} 
  \begin{umlseqdiag}
  
    \umlobject[no ddots, x=7]{Server} 
    \umlobject[no ddots, x=0]{CA}

    \begin{umlcallself}
      [side=left,
      padding=1,
      op={$\skServer, \pkServer = \Gen(1^K, n)$}]
      {Server}

    \begin{umlcallself}
      [side=left,
      padding=1,
      dt=4, 
      op={generate $CSR$  \makecell[l]{pk: $\pkServer$\\CN: \domain}}]
      {Server} 
    \end{umlcallself} 
    
    \begin{umlcall}
      [dt=4, 
      op={$CSR$}, 
      return={$cert_{server}$}] %
      {Server}{CA} 
      
      \begin{umlcallself}
      [padding=1,
      dt=2, 
      op={$\Sign$ CSR}]{CA} %
      \end{umlcallself} 
      
    \end{umlcall} 
     \begin{umlcallself}
      [side=left,
      padding=1,
      dt=4,
      op={$\skDeleg = \Del(\skServer, \epochform, \domain)$}]
      {Server} 
    \end{umlcallself} 
    
    \umlcreatecall[dt=4, no ddots, x=0]{Server}{CDN}
    
    \begin{umlcall}
      [padding=0,
      dt=3, 
      type=asynchron,
      op={Mutual Authentication}]{Server}{CDN} 
      \begin{umlcall}
        [padding=0,
        dt=4,
        op=$\skDeleg$]{Server}{CDN} 
      \end{umlcall}
      \begin{umlcall}
        [padding=0,
        op={$cert_{server}$}]{Server}{CDN} 
      \end{umlcall}
    \end{umlcall} 
    
  \end{umlcallself} 
  \end{umlseqdiag} 
\end{tikzpicture}
    \caption{TLS setup involving $\TBIDS$. The origin server creates a delegated key \skDeleg{} and shares it with the CDN. The secret key \skServer{} never leaves the origin server.}
    \label{fig:tbids:setup}
\end{figure}
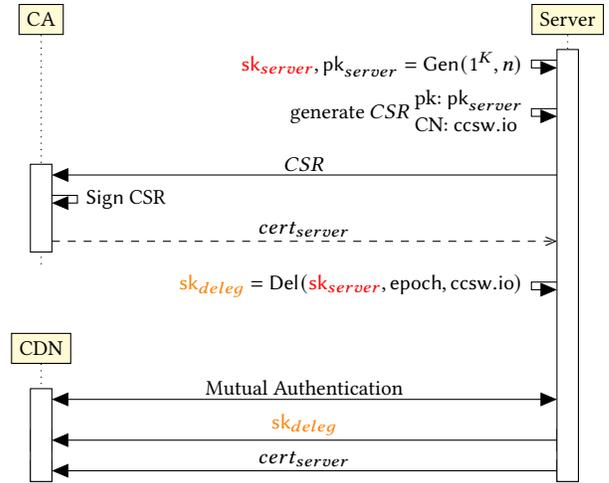

\begin{figure}[th]
    \centering
    \newcommand{\domain}{\textcolor{black}{\ensuremath{\mathsf{ccsw.io}}}}
\newcommand{\epochform}{\textcolor{black}{\ensuremath{\mathsf{epoch}}}}

\begin{tikzpicture} 
  \begin{umlseqdiag} 
  
    \umlobject[no ddots]{Client} 
    \umlobject[no ddots, x=7]{CDN}

    \begin{umlcall}
      [padding=0,
      op={ClientHello}, 
      ]{Client}{CDN}

    \begin{umlcall}
      [padding=1, 
      type=return, 
      op={Server Hello}]{CDN}{Client} 
    \end{umlcall} 
    
    \begin{umlcall}
      [padding=1, 
      dt=2,
      type=return, 
      op={Encrypted Extentions}]{CDN}{Client} 
    \end{umlcall} 
    
    \begin{umlcall}
      [padding=1, 
      dt=2,
      type=return, 
      op={Server Certificate ($cert_{server}$)}]{CDN}{Client} 
    \end{umlcall} 
    
    \begin{umlcallself}
      [padding=1, 
      dt=2,
      side=left,
      op={$\sigma=\Sign(\skDeleg, handshake)$}]{CDN} 
    \end{umlcallself}
    
    \begin{umlcall}
      [padding=1, 
      dt=4,
      type=return, 
      op={Certificate Verify ($\sigma$)}]{CDN}{Client} 
    \end{umlcall} 
    
    \begin{umlcall}
      [padding=1,
      dt=2,
      type=return, 
      op={Handshake Finished}]{CDN}{Client} 
    \end{umlcall} 
    
    \begin{umlcallself}
      [padding=1,
      dt=1,
      op={$epoch = getCurrentEpoch()$}]{Client} 
    \end{umlcallself} 
    
    \begin{umlcallself}
      [padding=1,
      dt=1,
      op={$\Verify(\pkServer, epoch, \domain, handshake, \sigma)$}]{Client} 
    \end{umlcallself} 
    
    \begin{umlcallself}
      [padding=1,
      dt=1,
      op={Verify PKI \& cert validity}]{Client} 
    \end{umlcallself} 
    
    \begin{umlcall}
      [padding=0, op={Client Finished}]{Client}{CDN} 
    \end{umlcall} 
    
    \end{umlcall}     
    
  \end{umlseqdiag} 
\end{tikzpicture}
    \caption{TLS handshake utilizing $\TBIDS$. The CDN uses the delegated key \skDeleg{} to sign the handshake, while the client uses the server's public key \pkServer{} and the \Verify~ function defined in \Cref{sm:tbids-ibs} to verify the signature.}
    \label{fig:tbids:handshake}
\end{figure}
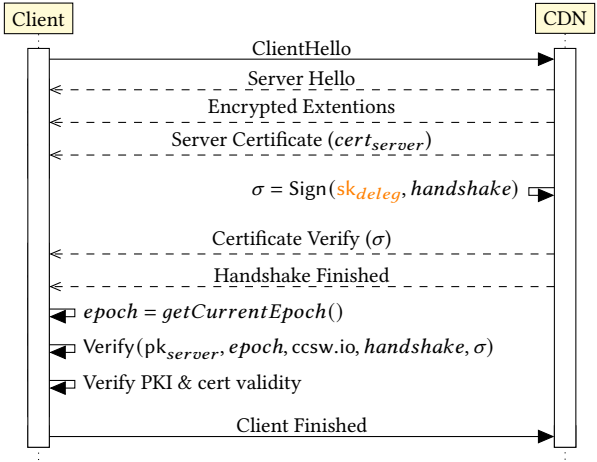

\subsection{Our Approach: Delegating the Secret Key}
\label{sec:approach}
To solve these issues, we use the  $\TBIDS$ scheme introduced in \Cref{sec:crypto} to enable an origin server to delegate the capability to sign on its behalf without sharing the secret key.
The setup is shown in \Cref{fig:tbids:setup}.
It is done by deriving a delegated key \skDeleg{} from the secret key \skServer{} by the origin server.
The server sends the delegated key to the CDN, which uses it for signing the TLS handshakes.
Clients use the public key from the origin server certificate to verify this signature, as shown in \Cref{fig:tbids:handshake}.
No changes to the TLS handshake are required.

A delegation derivation also involves the domain name of the origin server and a validity epoch.
Including the domain name can be used to limit the validity of the delegation to a subdomain.
This is further described in \Cref{sec:cryptoDef}.
The epoch restricts the validity period of the delegation, thus enabling short-lived delegations.
The origin server periodically derives a new delegation and sends it to the CDN over a mutually authenticated channel.
A CDN server always uses the delegation that is currently valid, thus signed for the current epoch.
To verify the signature, a client inputs the subdomain it connected to and the current epoch, together with the handshake’s hash.
If the signature verification process finds that those same input parameters produced the derived signature, it will assure its correctness.
The advantages of this approach are:
\begin{enumerate*}
    \item The CDN only holds a derived key, and the secret key never leaves the origin server. Therefore, the origin server remains in control of the private keys.
    \item The derived key only lasts for a limited amount of time, so the epoch limitation works as an inexpensive revocation alternative.
    \item If \fsTBIDS\ is used, the origin server's private key additionally enjoys forward-security properties.
    \item No changes to the TLS handshake are needed. Support for an additional signature algorithm is enough to integrate our approach.
\end{enumerate*}

As additional protection of the secret key, the origin server owner may use a dedicated keyserver, thus moving the key material to a different server than the data it serves.
The delegated key gets pushed to the CDN, allowing tighter network protection of the (key) server, since no incoming connections are necessary.
Since that is happening before the delegation becomes valid, there is no delay due to delegation or transport.
Besides, it is always possible to store the key material on secure hardware like a hardware security module (HSM).
Precomputation of delegated keys is also possible, allowing even more flexible key storage.

\begin{definition}[Epoch]
    We define an epoch as an integer, low-resolution timestamp, forming the identifier of a time period of length $\epochlength$. It is easily possible to derive the identifier from any timestamp or the clock, given by $\timestamp$. This timestamp represents the number of seconds since some $T_0$, e.g., the \textit{Unix epoch} or the certificate's \textit{NotBefore} field.
\end{definition}
An example of an algorithm is used in the TOTP protocol~\cite{DBLP:journals/rfc/rfc6238} and can be formulated as $epoch = \left\lfloor{\timestamp_{\UTC}}/{\epochlength}\right\rfloor$.

During signature verification, the client itself is deriving the current epoch from the clock.
Thus, since the CDN does not transmit the epoch, a malicious CDN cannot use an old delegation by convincing the client to use an old timestamp.
That ensures only delegations for the current time are valid.

\paragraph{Considerations on the Length of an Epoch}
Since the length of an epoch is neither part of the derivation itself nor a parameter part of the TLS handshake, it needs to be agreed on beforehand (e.g., defined in a corresponding standard document).
Similar to the validity period of an OCSP response, the length of an epoch can range from a few minutes up to a week \cite[Figure 8]{DBLP:conf/imc/ChungL0CLMMRSW18}.

On the one hand, the epoch length should not be too long.
While it is not possible to revoke a delegation after it has been issued and pushed out, revoking the certificate itself is possible.
By doing so, the delegation loses its trust since the certificate serving as the delegation's basis is no longer valid.
Since the revocation of a certificate often comes with some operational cost~\cite{DBLP:conf/imc/ChungL0CLMMRSW18,DBLP:conf/ndss/StarkHIJB12,DBLP:conf/pam/ZhuAH16}, it makes sense to set the length of an epoch to a short value.

On the other hand, end-user clients often do not have a precise clock time available~\cite{DBLP:conf/ccs/AcerSFFBDBST17}.
To counteract that, we propose to consider the epoch before and after the current one during the delegation's validation.
Depending on the implementation, the latter proposition requires up to two additional signature verifications.

\paragraph{Implication on Revocation}
While a delegated private key enables a third party (CDN) to use the corresponding certificate, it does not remove other PKI mechanisms currently in use.
A client is expected to verify the validity of the used certificate in addition to verifying the delegation.
As a consequence, revocation mechanisms like CRL and OCSP can still be used.
Furthermore, an epoch cannot expand the validity period of the certificate itself, although it provides an additional way to further limit the validity period.

\subsection{Changes to Existing TLS Implementations}
\label{sec:changes}
In this section, we describe changes that we had to make for implementing the scheme into an existing TLS library.

\paragraph{Certificate and Certificate Authority}
X.509 certificates are flexible enough to support public keys of new signatures schemes as long as an algorithm specifier (OID) is assigned, and an ASN.1 encoding is defined for the keys. Therefore, the library implementing \TBIDS\ support has to provide the handling of this data. Other specifying an OID and the encoding, no changes are required to the X.509 specification. For certificate authorities, the process of signing $\TBIDS$ public keys is not different from signing other public keys. For checking the Certificate Signing Requests (CSR), CAs require a $\TBIDS$ implementation to verify the self-signature.

\paragraph{Origin Server}
The server uses its secret key \skServer{} and a $\TBIDS$ tool to periodically derive a delegated key \skDeleg{}.
It then pushes the delegation to the CDN before the delegation becomes valid.
To do so, it may use a standard HTTP client and a mutually authenticated TLS connection.
Consequently, the latter does not involve changes to existing layers or software.

\paragraph{Content Delivery Network}
During setup and once before every epoch, the CDN receives a delegation key \skDeleg{} from the origin server.
Also, it retrieves the corresponding certificate, which is valid for the domain it serves.
During each new client connection, the CDN uses the received delegated key to sign the TLS handshake's hash.
Sending the resulting signature as \texttt{CertificateVerify} message to the client ensures the authentication of the connection.
To enable that, we integrated the $\TBIDS$ library in the TLS stack.
All protocol operations are part of the standard TLS handshake, and no additional ones are required outside.

\paragraph{Client}
We also extended the crypto library used by the TLS stack on the client with our $\TBIDS$ scheme.
Furthermore, we extended the TLS stack to support the handling of $\TBIDS$ public keys.
That allows the client to use standard X.509 certificates containing a public key of our scheme to verify the signature issued by the CDN.
Besides, it enabled the client to advertise support for $\TBIDS$ keys in the initial \texttt{ClientHello} message.
The required changes are limited to those commonly required to add a new signature algorithm.

In addition, the client carries out standard PKI checks to authenticate the server certificate, namely certificate chain check, revocation check, hostname validation, validity in time, and check corresponding to possible TLS extensions. No adaptions to those checks were required.

\begin{remark}[Implementation in TLS Libraries and Major Browsers]

Pairing libraries are widely available nowadays.
For libraries written in C or any language with C bindings available, relic~\cite{relic-toolkit} is a good choice to implement the $\TBIDS$ scheme. In particular, since relic implements optimization for the computation of pairing products.
Hence, integrating the scheme in OpenSSL or any of its forks as well as GnuTLS is easily possible.
Similarly, integration in major browsers such as Chrome and Firefox is thus also possible.
We also want to note that pairing libraries are also available for modern programming languages such as Rust.\footnote{For example, \url{https://github.com/zkcrypto/pairing}, retrieved May 04, 2020.}
Hence, even with Firefox's ongoing transition to Rust, integration of $\TBIDS$ is possible.
We also want to note the ongoing standardization effort of pairing-friendly curves to ensure compatibility between implementations~\cite{draft-pairing}.
\end{remark}

\subsection{Other Applications}
\label{sec:other}
\paragraph{TLS without CDNs}
In \Cref{sec:approach}, we described the application of $\TBIDS$ in the context of a CDN.
However, the scheme can also be used without a CDN.
For example, it is possible to use a delegated private key directly on the origin server or respective load balancers.
In that scenario, the same entity is in control of all servers.
The approach still enables extended protection for the private key by shielding it from the public network and reducing the risk of key compromise at the endpoints.

A key server can create a new delegation once per epoch and push it to the origin servers, and therefore does not need to be reachable from outside or even on all the time.
By doing so, not only the private key gets better protection, but the operator also gains the benefits of a short-lived certificate.
That means the damage in case of a compromise of the (derived) private key is limited to the length of the (short) epoch instead of the (longer) validity period of the certificate itself, even in cases where other revocation mechanisms are not available.

\paragraph{Restriction of connection types}
In our approach, we describe how time information can be used in an $\HIBS$ scheme to limit the validity period of a signature (and thus derived secret key) for $\TBIDS$.
It is also possible to use other information as an identity in a $\TBIDS$ scheme.
For example, one could put a constraint on the type of connection the TLS session is encapsulating.
That enables limiting the use of the delegation to, e.g., HTTP or SMTP connections.
Similar to the Extended Key Usage extension~\cite{DBLP:journals/rfc/rfc5280}, it is also possible to put further constraints on the usage of the derived key.

To achieve that, one needs to extend the modifications we describe in \Cref{sec:approach} and add a level that uses a normalized connection type specifier as identity.
For example, a delegation intended to be used by a webserver would be derived using `http` as identity.
A web browser would then use the same specifier `http` to verify the signature for HTTP connections.
Adding additional levels also allows CDNs to delegate keys on their own, i.e., CDN would receive a delegated key for a specific domain that can then be further derived for specific content types, specific servers, data centers, etc.

\paragraph{Other applications}
Additionally, we note that the validity dates in delegation schemes such as DeC or from \cite{DBLP:conf/openidentity/WagnerOM17} can be interpreted as more fine-grained epochs. Hence, if the end date of the validity period is set to the end of an epoch, DeC and other forms of delegation schemes can be interpreted as \TBIDS\ without $\fsEUFCMA$ security. Conversely, \TBIDS\ may find applications beyond the scope of TLS, e.g., for delegation in trust management.

\section{Evaluation}
\label{sec:Eval}

In this section, we evaluate an instantiation of the \TBIDS\ scheme introduced in \Cref{sec:crypto} based on \HIBS\ from BBG and our application into TLS described in \Cref{sec:integration}.
As expected, the cost of pairing evaluations dominates the performance.
Nevertheless, reasonable fast pairing implementations are available, and the performance penalty for additional features is below reasonable round trip times.

\subsection{Performance Measurements}
\label{sec:performance}

To understand the performance implications of \TBIDS, we implemented it in C based on relic~\cite{relic-toolkit} for runtime figures.
Furthermore, we implemented the \TBIDS\ scheme as Java library based on \ECCelerate~\cite{ECCelerate} and integrated it in the \ISASILK~\cite{ISASILK} stack to verify our claims concerning easy integration into TLS.\footnote{Implementations are available at \url{https://github.com/IAIK/TBIBS/}.}
For the C implementation, we target 128 bit security.
Hence we choose a pairing-friendly Barreto-Lynn-Scott~\cite{DBLP:conf/scn/BarretoLS02} curve with 381 bits (BLS-381) for the C implementation.
Following the recent security level estimations~\cite{DBLP:conf/mycrypt/MenezesS016,DBLP:journals/joc/BarbulescuD19}, this curve provides roughly 128 bit security.
For the Java implementation, we chose a variant of Barreto-Naehrig~\cite{DBLP:conf/sacrypt/BarretoN05} curve as BLS curves are not implemented in \ECCelerate.
As BN curves, we choose both curves with 256 and 461 bits. The latter provides 128 bits security, whereas the former provides around 100 bits of security, which we deem enough for key material that is recycled every three months as in the case of Let's Encrypt~\cite{DBLP:conf/ccs/AasBCDEFHHKRSW19}.

\paragraph{\TBIDS\ Evaluation}
Runtime-wise, $\TBIDS$ signing requires two exponentiations in $\G_1$ and one in $\G_2$, whereas verification requires three pairings.
In comparison, the standard signature schemes EdDSA and ECDSA require only a small number of operations in $\Z_p$ and one exponentiation in $\G$ for signing and two exponentiations in $\G$ for verification.
We present the performance figures for \TBIDS\ in comparison to OpenSSL's EdDSA and ECDSA implementation in \Cref{tab:algo}.
For OpenSSL, the numbers have been obtained by running \texttt{openssl speed}, whereas the numbers for \TBIDS\ were obtained by averaging over 10000 runs.
We performed all experiments on a Thinkpad T450s with an Intel Core i7-5600U CPU.
While the numbers suggest that \TBIDS\ is significantly slower, we note that the ECDSA implementation in OpenSSL is heavily optimized for the specific architecture.
If assembler optimizations are disabled, verification of ECDSA is close to \TBIDS\ signing, as expected from the number of group operations.

\begin{table}[t]
  \centering
  \caption{Runtime benchmarks in ms of ECDSA, EdDSA, and $\TBIDS$ from \Cref{sm:tbids-ibs}. The numbers for Ed25519 and ECDSA are from OpenSSL 1.1.1 with enabled optimizations.}
  \label{tab:algo}
  \begin{tabular}{lrr}
  \toprule
  Algorithm & $\Sign$ & $\Verify$ \\
  \midrule
  EdDSA (Ed25519)    & 0.05 & 0.14 \\
  ECDSA (sepc256r1) %
                    & 0.02 & 0.08 \\
  \TBIDS\ (BLS-381)  & 0.58 & 1.94 \\
\bottomrule
\end{tabular}
\end{table}

Regarding sizes of signatures and public keys, both EdDSA and ECDSA have signatures with two scalars in $\Z_p$, and the public key contains a group element.
For \TBIDS, public keys contain an element in $\G_2$ and signatures contain one element from $\G_1$ and $\G_2$ each.
Therefore, we can approximate its size as 64 bytes, 96 bytes, and 116 bytes if instantiated with BN-256, BLS-381, and BN-461, respectively, if point compression is enabled and as 128 bytes, 191 bytes, and 230 bytes, respectively, if not enabled.
An EdDSA or ECDSA public key only uses 32 bytes with point compression enabled and 64 bytes without.
Furthermore, the signature in the latter only sums up 64 bytes (two scalars in $\Z_p$), while the \TBIDS's signature needs one element of both groups $\G_1$ and $\G_2$, which results in 96, 143, and 173 bytes, respectively, with point compression and 192 bytes, 286 bytes, 345 bytes, respectively without point compression.

\paragraph{TLS Evaluation}
To measure the performance of \TBIDS\ in the use case of TLS, we implemented a \Java\ Signature library based on \JCA\ guidelines.
The \IAIK\ \ECCelerate\ library served as a backend for cryptographic elliptic curve operations.
Furthermore, the \ISASILK\ library, supporting the Java Cryptography Architecture, was adapted to the \TBIDS\ signature library.
Then, a demo client and server were implemented with this TLS library.
The measurements for the other algorithms were also performed using the same \Java\ libraries.

For benchmarking the \TBIDS\ algorithm and TLS adaption, we used OpenJDK’s Java Microbenchmark Harness (JMH) library.
Furthermore, we used Linux’s traffic control queueing disciplines for simulating network latency on our TLS benchmarks.
The latency we introduced amount to 10ms, with a variance of 1ms normal distributed.
That results in a round trip time of 20ms, a typical delay for a home’s network with Wireless and glass fiber connecting to Google (c.f.~\cite{verizon}).
Since all benchmarks have been performed on a Thinkpad T450s with an Intel Core i7-5600U, the absolute numbers may not be comparable with performance on a server CPU.

\Cref{tab:tls} presents benchmarks showing the cost of a TLS handshake using ECDSA, EdDSA, and \TBIDS.
We can see that cryptographic performance plays a much smaller factor in the overall performance, as the slowdown is smaller than a factor of 2. For instance, with BN-256, we only experience a factor of 1.2 as a slowdown.
Further optimization, such as precomputation (discussed in \Cref{sec:practical-precomputation}) and assembler implementation, may reduce the performance gap even further. Bandwidth-wise, \Cref{tab:bytes} shows that, in principle, \TBIDS\ needs twice the amount of bytes than the other schemes.
On the other hand, the size does not exceed any supported TLS size limits.

\begin{table}[t]
  \centering
  \caption{Benchmarks of TLS 1.3 handshakes using RSA, ECDSA, EdDSA, and \TBIDS\ (one operation is a full handshake) and sizes of the packets sent by the client (C) and the server (S).}
  \label{tab:tls} \label{tab:bytes}
  \begin{tabular}{lrrrrr}
    \toprule
    Algorithm         & ops/s        & s/ops  & bytes (C) & bytes (S) \\
    \midrule
    RSA (2048 bit)    &        8.942 &  0.112 & 643 & 1773 \\
    ECDSA (secp256r1) &        9.254 &  0.108 & 643 & 1385 \\
    EdDSA (ed25519)   &        9.198 &  0.109 & 643 & 1330 \\
    \TBIDS\ (BN-256)  &        7.703 &  0.130 & 645 & 2769 \\
    \TBIDS\ (BN-461)  &        4.598 &  0.217 & 645 & 3915 \\
    \bottomrule
  \end{tabular}
\end{table}

\paragraph{Discussion}
As depicted in \Cref{tab:tls}, applying \TBIDS\ to TLS slows the handshake by a factor of 1.2 with the smaller BN curves and a factor of 2 with the larger curves.
That applies as long as we assume a reasonable, inter-continental latency.
If the latency is higher than 10ms, the performance difference decreases further. Furthermore, we note that highly optimized assembler code, as it can be found in OpenSSL for ECDSA, would further reduce the performance gap.
As the relic-based implementation shows, verification performance can be far beneath typical network latencies.
Also, smaller and embedded devices may benefit from hardware acceleration~\cite{DBLP:conf/ches/UnterluggauerW14}.
Besides, authentication using certificates is only required once.
Every further handshake to the CDN may use session resumption, thus omitting any $\TBIDS$ related operations.
Since a client may connect more than once to the CDN until all data is loaded, the difference becomes imperceptible.

Comparing performance to Keyless SSL, our approach has the advantage that means for authentication are pushed in prior and do not have to be pulled just in time for every handshake.
Further, while the latencies between client and CDN might be subtle, they can be pretty high on CDN to key server side, as they might, e.g., be geographically far apart~\cite{verizon}.

\subsection{Qualitative Evaluation}
\label{sec:qual_eval}

The most prominent advantage of our approach is that delegated keys are bound to a specific epoch.
That means if a delegated key gets compromised, it only affects the security of that epoch. In subsequent epochs, the unforgeability guarantees are restored, and thus proper authentication can be performed without entirely revoking and recycling the origin server's key material.

Besides, the origin server's private key can be equipped with forward-security features; thereby, the approach also offers mitigation against the master key's compromise. If the origin server immediately updates its private key to the next epoch after the delegations for the current epoch, compromise of that key does not impact the authenticity of connections established in the current epoch.

Another exciting aspect is the lack of modifications needed to run the \TBIDS\ scheme in TLS.
The reason for it is that our short-lived delegation infrastructure exclusively touches the cryptographic part of the TLS stack.
Because of that, PKI's mechanisms, such as revocation, still work without any changes as our proposal does not touch them.
Consequently, there is no need for a new multi-party security assumption to be proven since no changes to the TLS protocol per se are introduced.
Therefore, the (S)ACCE proof from Jager et al.~\cite{DBLP:conf/crypto/JagerKSS12} on TLS persists.

For the widespread adoption of a new signature scheme like $\TBIDS$, we would need approval by the CA/Browser Forum. Such would demonstrate a wide acceptance by CAs, browser developers, and other internet stakeholders. DeC, on the other hand, only needs an additional TLS extension, which is advantageous when deploying.

Besides, our approach can offer an advantage for steep delegation hierarchies. Assuming we want to delegate by several levels $h$:
Most certificate-based solutions would verify a chain of certificates resulting in $h$ verifications.
However, our approach would transparently support the additional levels while only consuming an additional group operation per delegation level $h$.

\paragraph{Systematic Analysis}
Finally, we characterize our approach by applying the analysis framework published recently by Chuat et al.~\cite{chuat2019sok}. In \Cref{tab:comp}, we recall their characterization of the relevant approaches but omit the (damage-free) CA revocation criteria, which are not relevant for delegation approaches and the delegation criterion as it is inherently satisfied by all approaches. We extend their characterization with our \TBIDS-based approach (and the STAR delegation approach).
\newcommand{\rot}[1]{\rotatebox[origin=c]{90}{#1}}
\begin{table*}[t]
\caption{Evaluation of schemes with respect to 16 criteria of Chuat et al.'s 19 criteria framework~\cite{chuat2019sok}. \priority{100} offers the benefit; \priority{50} partially offers benefit; \priority{0} does not offer the benefit.}
  \label{tab:comp}
\begin{tabular}{l|c|cc|ccc|ccc|cccccc|r}
\toprule
Scheme                                   &
\rot{Avoids Full Delegation}             &
\rot{Supports leaf revocation}           &
\rot{Supports autonomous leaf rev.}      &
\rot{Supports domain-based policies}     &
\rot{No trust-on-first-use required}     &
\rot{Preserves user privacy}             &
\rot{No increased page-load delay}       &
\rot{Low burden on CAs}                  &
\rot{Reasonable logging overhead}        &
\rot{Non-proprietary}                    &
\rot{No special hardware required}       &
\rot{No extra CA involvement}            &
\rot{No browser-vendor involvement}      &
\rot{Server compatible}                  &
\rot{Browser compatible}                 &
\rot{No out-of-band communication}       \\
\midrule
Private key sharing &
\priority{0}   & \priority{0}   & \priority{0}   & \priority{0}   & \priority{100} & \priority{100} & \priority{100} & \priority{100} & \priority{100} & \priority{100} & \priority{100} & \priority{100} & \priority{100} & \priority{100} & \priority{100} & \priority{100}\\
Cruiseliner certificates~\cite{DBLP:conf/ccs/CangialosiCCLMM16,DBLP:conf/sp/LiangJDLWW14} &
\priority{0}   & \priority{0}   & \priority{0}   & \priority{0}   & \priority{100} & \priority{100} & \priority{100} & \priority{100} & \priority{100} & \priority{100} & \priority{100} & \priority{0}   & \priority{100} & \priority{100} & \priority{100} & \priority{100}\\
Name const. certificates~\cite{DBLP:conf/sp/LiangJDLWW14} &
\priority{100} & \priority{0}   & \priority{0}   & \priority{100} & \priority{100} & \priority{100} & \priority{100} & \priority{100} & \priority{0}   & \priority{100} & \priority{100} & \priority{0}   & \priority{100} & \priority{100} & \priority{50}  & \priority{100}\\
DANE-based~\cite{DBLP:conf/sp/LiangJDLWW14} &
\priority{100} & \priority{0}   & \priority{0}   & \priority{0}   & \priority{100} & \priority{100} & \priority{0}   & \priority{100} & \priority{100} & \priority{100} & \priority{100} & \priority{100} & \priority{100} & \priority{100} & \priority{0}   & \priority{100}\\
SSL splitting~\cite{DBLP:journals/cn/Lesniewski-LaasK05} &
\priority{100} & \priority{0}   & \priority{0}   & \priority{0}   & \priority{100} & \priority{100} & \priority{0}   & \priority{100} & \priority{100} & \priority{100} & \priority{100} & \priority{100} & \priority{100} & \priority{0}   & \priority{100} & \priority{100}\\
Keyless SSL~\cite{cloudflare-keyless-ssl-details} &
\priority{100} & \priority{0}   & \priority{0}   & \priority{0}   & \priority{100} & \priority{100} & \priority{0}   & \priority{100} & \priority{100} & \priority{100} & \priority{100} & \priority{100} & \priority{100} & \priority{0}   & \priority{100} & \priority{100}\\
STAR Delegation~\cite{starDelegate} &
\priority{100} & \priority{100} & \priority{0}   & \priority{100} & \priority{100} & \priority{100} & \priority{100} & \priority{0}   & \priority{0}   & \priority{100} & \priority{100} & \priority{0}   & \priority{100} & \priority{0}   & \priority{100} & \priority{100}\\
Proxy certificates~\cite{DBLP:journals/corr/abs-1906-10775} &
\priority{100} & \priority{50}  & \priority{50}  & \priority{100} & \priority{100} & \priority{100} & \priority{100} & \priority{100} & \priority{100} & \priority{100} & \priority{100} & \priority{100} & \priority{100} & \priority{50}  & \priority{0}   & \priority{100}\\
DeC~\cite{DelegatedCredentials} &
\priority{100} & \priority{100} & \priority{100} & \priority{50}  & \priority{100} & \priority{100} & \priority{100} & \priority{100} & \priority{100} & \priority{100} & \priority{100} & \priority{50}  & \priority{100} & \priority{0}   & \priority{0}   & \priority{100}\\
\TBIDS      &
\priority{100} & \priority{100} & \priority{100} & \priority{50}  & \priority{100} & \priority{100} & \priority{100} & \priority{100} & \priority{100} & \priority{100} & \priority{100} & \priority{50}  & \priority{100} & \priority{50}  & \priority{50}  & \priority{100}\\
\bottomrule
\end{tabular}
\end{table*}
We also discuss the fulfillment of the criteria below.
\begin{description}
\item[Avoids Full Delegation] This criterion is fulfilled for any approach which keeps the origin server in sole control of the private key.

\item[Supports leaf revocation] DeC and \TBIDS\ technically do not satisfy this benefit directly since revoking a delegation at the chain of trust's end is impossible. However, thanks to their short-lived characteristic, domain owners are able to invalidate a delegation on key compromise or similar eventualities swiftly. Chuat et al. characterize proxy certificate based approaches as partially fulfilling this criterion since they can also be short-lived.

\item[Supports autonomous revocation] With both, DeC and \TBIDS\, it is possible to perform revocation autonomously, i.e., independent from a CA, browser vendor, or log. For both, the domain owner can decide independently to stop the short-lived delegation by ending the distribution of credentials, respectively delegated keys. As before, proxy certificates partially fulfill this criterion if they are used in a short-lived manner.

\item[Support domain-based policies] While name constraints and proxy certificates can specify such policies in the certificate, DeC and \TBIDS\ partially fulfill this criterion since their semantics are limited (only a time component is supported in its standard versions).

\item[No trust-on-first-use required] None of the approaches requires trust-on-first-use since the PKI remains unchanged.

\item[Preserves user privacy] None of the delegation approaches leaks domain-related data to a third party. Hence all of them fulfill this criterion.

\item[No increased page-load delay] Chuat et al. consider this criterion fulfilled if none or small processing delays incur. SSL splitting, KeylessSSL, and the DANE-based approach all require additional round-trips and therefore do not fulfill it. The \TBIDS-based approach only incurs a small overhead, as demonstrated in \Cref{tab:algo}.

\item[Low burden on CAs] Approaches fulfill this requirement if no or low operational overhead incurs for Certificate Authorities. Only STAR delegations do not satisfy this criterion since the CA has to support the ACME protocol and trust the CDN.

\item[Reasonable logging overhead] The approaches putting heavy pressure on certificate transparency (CT) logs~\cite{DBLP:journals/queue/Laurie14} are name constraint certificates and STAR delegations. For name constraint certificates, every domain owner may issue an arbitrary number of certificates that are recorded on the CT logs. STAR delegations produce a growing number of short-lived certificates for each delegation. Neither DeC nor \TBIDS\ require delegations to be logged.

\item[Non-proprietary] All of the approaches are open and neither restricted nor controlled by a third party.

\item[No special hardware is required] Although cryptographic operations may always benefit from specialized hardware, it is not necessary for deployment. Hence, none of the approaches necessarily need specialized hardware.

\item[No extra CA involvement] Chuat et al. characterize DeC as partially satisfying this criterion since the CA needs to include an extension in the end-entity certificate. For \TBIDS, the CA needs to accept \TBIDS\ public keys, hence we also consider our approach as partially fulfilling this criterion. However, name constraint, Cruiseliner, and STAR delegation certificates do not satisfy this criterion as a significant part of the delegation process is managed by the CA.

\item[No browser-vendor involvement] None of the approaches needs active browser-vendor participation.

\item[Server compatibility] We consider \TBIDS\ as partially satisfying this criterion since adding support for it on the server side is mostly a matter of adding the signature scheme implementation to the TLS stack. Other approaches including DeC and KeylessSSL need more involved changes in the TLS stack, and hence does not satisfy this constraint. The same holds for true for STAR delegations since it requires the origin server to approve and pass CSR to the ACME CA while authenticating as the domain owner.

\item[Browser compatibility] Again, we give partial points to \TBIDS\ since only an implementation of the signature scheme needs to be added to the TLS stack to ensure support for the approach.

\item[No out-of-band communication] None of the approaches uses a separate channel or communicates with a third party server that would not otherwise be contacted by the client.

\item[Damage-free CA-certificate revocation] Our approach does not directly satisfy this criterion. However, since we do not touch the PKI, a combination with measures such as PKISN~\cite{DBLP:conf/eurosp/SzalachowskiCP16} yields a generic solution that is applicable to any approach that does not satisfy this criterion itself.
\end{description}

\subsection{Conclusion}
In this paper we showed that signature schemes with short-lived delegation provide an interesting alternative to current delegation practices in the TLS ecosystem.
In specific, we presented a forward-secure alternative to the approach used by Delegated Credentials.
We used \TBIDS\ to derive short lived keys that can sign directly on behalf of the master key.
Preservation of standard PKI infrastructure and forward-security are the main features of the approach and come with only small performance overheads, which we demonstrated to be practical.

\begin{acks}
This work was supported by the
  \grantsponsor{EUH2020}{European Union's Horizon 2020 research and innovation programme}{https://ec.europa.eu/programmes/horizon2020/en} under grant agreements \numero~\grantnum{EUH2020}{871473} (KRAKEN) and \numero~\grantnum{EUH2020}{783119} (SECREDAS).
\end{acks}

\bibliographystyle{ACM-Reference-Format}
\bibliography{bib,dblp}

\ifdefined\full
\appendix

\section{Signature Schemes}
\label{app:sigs}

For completeness, we recall the standard definition of a signature scheme.
\begin{definition}
  A signature scheme $\Sigma$ consists of the PPT algorithms $(\Gen, \Sign,\Verify)$, which are defined as follows:
\begin{description}
  \item[$\Gen(1^\secpar)\colon$] On input security parameter $\secpar$ outputs a signing key $\sk$ and a verification key $\pk$ with associated message space $\Msg$.
  \item[$\Sign(\sk, \msg)\colon$] On input, a secret key $\sk$ and a message $\msg \in \Msg$, outputs a signature $\sigma$.
  \item[$\Verify(\pk,\msg,\sigma)\colon$] On input a public key $\pk$, a message $\msg \in \Msg$ and a signature $\sigma$, outputs a bit $b$.
\end{description}
\end{definition}
We note that for a signature scheme many independently generated public keys may be with respect to the same parameters $\pp$, e.g., some elliptic curve group parameters. In such a case we use an additional algorithm $\Setup$, and $\pp \gets \Setup(1^\secpar)$ is then given to $\Gen$. We assume that a signature scheme satisfies the usual (perfect) correctness notion, i.e. for all security parameters $\secpar \in \N$, for all $(\pk, \sk) \gets \Gen(1^\secpar)$, for all $m \in \Msg$, we have that
\[\Pr\mleft[\Verify(\pk, \msg, \Sign(\sk, \msg)) = 1\mright] = 1\text{.}\]
Below, we recall the standard existential unforgeability under adaptively chosen message attacks ($\EUFCMA$ security) notion, which requires that even with access to a signing oracle an adversary is unable to forge signature on message that have not been queried.
\begin{experiment}[h]
  \begin{flushleft}
  The experiment has access to the following oracles:
  \begin{description}
    \item[$\Sign'(\sk,\msg)\colon$] This oracle computes $\sigma \gets \Sign(\sk, \msg)$, adds $\msg$ to $\mathcal{Q}$, and returns $\sigma$.
  \end{description}
\end{flushleft}
  \textbf{Experiment} $\ExpSigEUFCMA{\Sigma,\advA}(\secpar)$
  \begin{algorithmic}
    \State $(\pk,\sk) \gets \Gen(1^\secpar)$
    \State $(\msg^*, \sigma^*) \gets \advA^{\Sign'(\sk, \cdot)}(\pk)$
    \State if $\Verify(\pk, \msg^*, \sigma^*) = 0$, return 0
    \State if $\msg^* \in \mathcal{Q}$, return 0
    \State return 1
  \end{algorithmic}
  \caption{The $\EUFCMA$ experiment for a signature scheme $\Sigma$.}
  \label{exp:sig-eufcma}
\end{experiment}
\begin{definition}[$\EUFCMA$]
  For any PPT adversary $\advA$, we define the advantage in the $\EUFCMA$ experiment $\ExpSigEUFCMA{\Sigma,\advA}$ (cf. \Cref{exp:sig-eufcma}) as
  \begin{align*}
    \AdvSigEUFCMA{\Sigma,\advA}(\secpar):=\Pr\left[\ExpSigEUFCMA{\Sigma,A}(\secpar) = 1\right] \text{.}
  \end{align*}
  A signature scheme $\Sigma$ is $\EUFCMA$-secure, if $\AdvSigEUFCMA{\Sigma,\advA}(\secpar)$ is a negligible function in $\secpar$ for all PPT adversaries $\advA$.
\end{definition}
\fi

\end{document}